\documentclass[10pt]{amsart}
\usepackage{epsf,latexsym,amsmath,amssymb,amscd,datetime}
\usepackage{amsmath,amsthm,amssymb,enumerate,eucal,url,calligra,mathrsfs}

\usepackage{blkarray} 

\usepackage{tikz,scalerel}
\usetikzlibrary{arrows}

\usepackage{imakeidx}     

\usepackage[outdir=./]{epstopdf}

\usepackage{graphicx}
\usepackage{color}
\newenvironment{jfnote}{ \bgroup \color{blue} }{\egroup}

\newcommand{\red}{\color[rgb]{1.0,0.2,0.2}} 
\newcommand{\blue}{\color[rgb]{0.2,0.2,1.0}} 

\newcommand{\oldStuff}[1]{}

\newcommand{\Coord}{{\rm Coord}}
 
\newcommand{\og}{{\scriptscriptstyle \le}}
\newcommand{\Bg}{{\scalebox{1.0}{$\!\scriptscriptstyle /\!B$}}}

\newcommand{\cert}{\xi}
 
\DeclareMathOperator{\SHom}{\mathscr{H}\text{\kern -3pt {\calligra\large om}}\,}

\IfFileExists{my_xrefs}{\input my_xrefs}{}

\DeclareMathOperator{\ViSu}{VisSub}

\newcommand{\naturals}{{\mathbb N}}

\newcommand{\Eor}{E^{\mathrm{or}}}
\newcommand{\mec}[1]{{\bf #1}}	
\newcommand{\bec}[1]{{\boldsymbol #1}}	

\DeclareMathOperator{\trace}{Trace}
\DeclareMathOperator{\Trace}{Trace}

\usepackage{mathrsfs}
\usepackage{amssymb}
\usepackage{dsfont}
\usepackage{verbatim}
\usepackage{url}

\newcommand{\Edir}{E^{\mathrm{dir}}}

\theoremstyle{plain}
\newtheorem{theorem}{Theorem}[section]
\newtheorem{lemma}[theorem]{Lemma}
\newtheorem{proposition}[theorem]{Proposition}
\newtheorem{corollary}[theorem]{Corollary}

\theoremstyle{definition}
\newtheorem{definition}[theorem]{Definition}

\newtheorem{xca}{Exercise}[section]

\newcommand{\isom}{\simeq} 

\newcommand{\ignore}[1]{}

\newcommand{\reals}{{\mathbb R}}

\newcommand{\integers}{{\mathbb Z}}

\newcommand{\complex}{{\mathbb C}}

\newcommand\EE{\mathbb{E}}

\newcommand\II{\mathbb{I}}

\DeclareMathAlphabet{\mathcal}{OMS}{cmsy}{m}{n}

\newcommand\cC{\mathcal{C}}

\newcommand\cO{\mathcal{O}}
\newcommand\cP{\mathcal{P}}

\newcommand\cR{\mathcal{R}}
\newcommand\cS{\mathcal{S}}
\newcommand\cT{\mathcal{T}}

\DeclareMathOperator{\Prob}{Prob}

\DeclareMathOperator{\VLG}{VLG}

\DeclareMathOperator{\Line}{Line}

\DeclareMathOperator{\SNBC}{SNBC}

\DeclareMathOperator{\snbc}{snbc}

\def\from{\colon}

\def\isom{\simeq}

\def\eqdef{\overset{\text{def}}{=}}

\DeclareMathOperator{\ord}{ord}

\def\implies{\Rightarrow}

\DeclareRobustCommand
  \rddots{\mathinner{\mkern1mu\raise\p@
    \vbox{\kern7\p@\hbox{.}}\mkern2mu
    \raise4\p@\hbox{.}\mkern2mu\raise7\p@\hbox{.}\mkern1mu}}

\newcommand\xhookrightarrow[2][]{\ext@arrow 0062{\hookrightarrowfill@}{#1}{#2}}
\def\hookrightarrowfill@{\arrowfill@\lhook\relbar\rightarrow}

\usepackage{relsize}
\usepackage{tikz}
\usetikzlibrary{matrix,arrows,decorations.pathmorphing}
\usepackage{tikz-cd}
\usetikzlibrary{cd}

\usepackage[pdftex,colorlinks,linkcolor=blue]{hyperref}

\newcommand{\myDeleteNote}[1]{{\blue myDeleteNote:}\ {\red #1}}

\begin{document}

\title[Relativized Alon Conjecture III] 
{On the Relativized Alon Second Eigenvalue
Conjecture III: Asymptotic Expansions for Tangle-Free
Hashimoto Traces}

\author{Joel Friedman}
\address{Department of Computer Science,
        University of British Columbia, Vancouver, BC\ \ V6T 1Z4, CANADA}
\curraddr{}
\email{{\tt jf@cs.ubc.ca}}
\thanks{Research supported in part by an NSERC grant.}

\author{David Kohler}
\address{Department of Mathematics,
        University of British Columbia, Vancouver, BC\ \ V6T 1Z2, CANADA}
\curraddr{422 Richards St, Suite 170, Vancouver BC\ \  V6B 2Z4, CANADA}
\email{{David.kohler@a3.epfl.ch}}
\thanks{Research supported in part by an NSERC grant.}

%
\date{\today}



\begin{abstract}

This is the third in a series of articles devoted to showing that a typical
covering map of large degree to a fixed, regular graph has its new adjacency
eigenvalues within the bound conjectured by Alon for random regular graphs.

In this paper we consider random graphs that are random covering graphs 
of large
degree $n$ of a fixed base graph.  We prove the existence of asympototic
expansion in $1/n$ for the expected value of the 
number of strictly non-backtracking closed walks of length $k$
times
the indicator function that
the graph is free of certain {\em tangles};
moreover, we prove
that the coefficients of these expansions
are ``nice functions'' of $k$, 
namely approximately equal to a sum of polynomials
in $k$ times exponential functions of $k$.

Our results use the methods of Friedman used to resolve
Alon's original conjecture, combined with
the results of Article~II in this series of articles.  One simplification in
this article over the previous methods of Friedman is that the
``regularlized traces'' used in this article, which we call {\em certified
traces}, are far easier to define and work with than the 
previously utilized {\em selective traces}.

\end{abstract}

\maketitle
\setcounter{tocdepth}{3}
\tableofcontents

\newcommand{\sePrelimProofs}{17}

\section{Introduction}

This is the third article in a series of six devoted to proving
two main results of a generalization of Alon's second eigenvalue
conjecture.  

Alon's original conjecture 
(\cite{alon_eigenvalues}, Conjecture~5.1) is that for fixed integer $d$,
and real $\epsilon>0$,
the second adjacency matrix eigenvalue of
``most'' $d$-regular graphs on $n$ vertices is at most
$2\sqrt{d-1}+\epsilon$.
Graphs with so small a second eigenvalue provably ``expand'' sets
of vertices of certain sizes when passing to their set of neighbours 
(e.g., Section~1 of \cite{alon_eigenvalues} and
also \cite{alon_milman,tanner,dodziuk}); 
see \cite{hoory_linial_wigderson}
for many other applications of expanders and the
explicit construction of expanders \cite{lps,margulis,morgenstern}.
We do not know the motivation for Alon's conjecture, although
Section~4 of \cite{alon_eigenvalues} states the result of 
Alon and Boppana \cite{nilli} that $2\sqrt{d-1}+o_n(1)$ is a lower bound for
all $d$-regular graphs on $n$ vertices.
Alon's conjecture with weaker upper bounds appears in
\cite{broder,friedman_kahn_szemeredi,friedman_random_graphs},
and was finally proven in \cite{friedman_alon}.
Both the counting argument of Kahn and Szemeredi 
\cite{friedman_kahn_szemeredi}, 
and the trace methods of Broder \cite{broder} and
Friedman \cite{friedman_random_graphs,friedman_alon} give the
same lower bound on the absolute value of all the negative eigenvalues.

In this series of articles we study a generalization of Alon's 
conjecture, where one fixes a ``base'' graph $B$, and studies
random covering maps of degree $n$ to $B$.
The {\em new eigenvalues} \cite{friedman_relative_boolean,friedman_relative}
of the adjacency matrix of the covering graph are those not
arising from (and therefore orthogonal to)
the pullback of eigenfunctions on $B$; \cite{friedman_relative}
Section~5
conjectures that the largest absolute of the new eigenvalues
of ``most'' such covering maps
is $\rho(B)+o_n(1)$ where $\rho(B)$ is the spectral radius of 
the adjacency matrix of the universal cover of $B$
(and $\rho(B)=2\sqrt{d-1}$ if $B$ is $d$-regular).
This conjecture generalizes Alon's conjecture, which is the
special case where $B$ consists of one vertex.
This conjecture was proven with a weaker upper bound in
\cite{friedman_relative}, by adapting the trace methods of
Broder-Shamir \cite{broder} to this more general situation;
counting methods improved these results in
\cite{lubetzky,a-b} for regular base graphs, 
as did trace methods \cite{linial_puder,puder}.
The conjecture was proven when $B$ is a regular graph 
in \cite{friedman_kohler}, and with a much simpler type of trace method,
but involving more sophisticated probabilistic methods, by
Bordenav\'e
\cite{bordenave}.  Recently Bordenav\'e and Collins
\cite{bordenave_collins2019}
have proven spectral bounds for a large class of sums of random
permutation matrices, 
and their result resolves the entire relativized
Alon conjecture, where the base graph need not
be regular.
The point of this series of
articles is to ``factor'' the proof in \cite{friedman_kohler} into
independent parts; this factorization includes a number of 
significant
simplifications to the methods in \cite{friedman_random_graphs,friedman_alon}.

The above generalization of Alon's conjecture is also a
{\em relativization} of this conjecture (now a theorem)
in the sense that one
associates with Grothendieck, in that this generalization is
a theorem regarding covering morphisms in categories where the special
case regarding morphisms to a final element of the category
(i.e. graphs with one vertex)
specialize to the original conjecture (which is viewed as a
conjecture regarding objects in these categories).

The study of relative expanders of fixed degree two \cite{bilu}
have lead to a number
of exciting recent results \cite{mssI,mssIV}, proving the
existence of relatively Ramanujan bipartite of all degrees and all even
number of vertices; see also \cite{hall_puder_sawin}.

Let us briefly describe the main technical result of this article,
in rough terms.

In this article we fix a graph $B$ (not necessarily regular)
and consider a family
$\{\cC_n(B)\}$ of random covering map of degree $n$ to $B$
(defined for at least an infinite set of integers, $n$)
that is {\em algebraic} in a sense we will define.
The main and most difficult result in this paper concerns the function
$$
f(k,n)=\EE_{G\in\cC_n(B}[ \II_{{\rm TangleFree}(\ge\nu,<r)}(G) 
\Trace(H_G^k) ] ,
$$
where $H_B$ is the Hashimoto matrix (also called the non-backtracking
matrix) of $G$, and $\II_{{\rm TangleFree}(\ge\nu,<r)}$ is the
indicator function of a graph being free of certain {\em tangles}.
We prove that such a function has an asymptotic expansion
$$
c_0(k) + c_1(k)/n + \cdots + c_{r-1}(k)/n^{r-1} + O_k(1/n^r),
$$
for $k \ll n^{1/2}$, where $c_0(k)$ is explicitly computed
to within a small error term,
and $c_1,\ldots,c_{r-1}$ are certain well-behaved functions
of $k$.
In Article~IV we study such a result purely in terms of 
probability theory to prove that the results in this article
can be used to control the eigenvalues of $H_G$, which in
Article~V allows us to prove the relativized Alon conjecture 
when $B$ is $d$-regular, and in Article~VI is used to determine
probability that new adjacency eigenvalues exceeding 
$2\sqrt{d-1}+\epsilon$ occur in $\cC_n(B)$ to within a multiplicative
constant (with depending on $\epsilon$ but not on $n$) provided
that $B$ is regular and {\em Ramanujan}.

This article uses the methods of \cite{friedman_alon} but makes
some simplifying modifications, the most significant of which is the
introduction of {\em certified traces} that replace the 
{\em selective traces} of \cite{friedman_alon}.
It is not only simpler to define certified traces, but their
properties (and proofs using certified traces) can be
``factor'' into more independent parts.

The theorems in this paper requires the results of Article~II.
We assume that the reader is familiar with definitions in Article~I.
For ease of reading, we have summarized all definitions that we need
in Section~\ref{se_defs_review}; however, to understand the motivation
and the subtleties of these definitions we refer the reader to Article~I.

The rest of this article is organized as follows.
After reviewing the definitions needed in Section~\ref{se_defs_review},
we state the main results in this article in 
Section~\ref{se_main_theorems}.
We review the theorems needed from Article~II in 
Section~\ref{se_art_II_theorems}.
In Section~\ref{se_certified} we explain our strategy to prove our
main theorems, and state two more general theorems from which our
main theorems easily follow.
In Section~\ref{se_cert_proof}
we define $(\ge\nu,<r)$-tangles and certified traces, and discuss how they are
related to traces of the Hashimoto matrix of a graph.
In Section~\ref{se_cert_proof} we prove an asymptotic expansion
theorem regarding certified traces; in the process we prove that
each certified trace has a finite number of ``certificates,''
which is why it is easy to use them in proofs
(much easier that the {\em selective traces} in
\cite{friedman_alon}).
In Section~\ref{se_finiteness_min_tang} we prove that there are 
a finite number of minimal $(\ge\nu,<r)$-tangles for any fixed
$\nu,r$ with $\nu>1$; this was done in \cite{friedman_alon}, but
we correct a minor error there.
In Section~\ref{se_indicator} we define a series of approximations to
the indicator function of a graph having $(\ge\nu,<r)$-tangles,
and give their basic properties.
In Sections~\ref{se_ind_cert_proof}
and Section~\ref{se_ind_proof} we prove the generalizations of the
main theorems in this paper stated in Section~\ref{se_certified},
and in Section~\ref{se_finish_proof} we prove our
main theorems using these generalizations.

\section{Review of the Main Definitions}
\label{se_defs_review}

We refer the reader to Article~I for the definitions used in this article,
the motivation of such definitions, and an appendix there that lists all the
definitions and notation.
In this section we briefly review these definitions and notation. 

\subsection{Basic Notation and Conventions}
\label{su_very_basic}

We use $\reals,\complex,\integers,\naturals$
to denote, respectively, the
the real numbers, the complex numbers, the integers, and positive
integers or
natural numbers;
we use $\integers_{\ge 0}$ ($\reals_{>0}$, etc.)
to denote the set of non-negative
integers (of positive real numbers, etc.).
We denote $\{1,\ldots,n\}$ by $[n]$.

If $A$ is a set, we use $\naturals^A$ to denote the set of
maps $A \to \naturals$; we will refers to its elements as
{\em vectors}, denoted in bold face letters, e.g., $\mec k\in \naturals^A$
or $\mec k\from A\to\naturals$; we denote its {\em component}
in the regular face equivalents, i.e., for $a\in A$,
we use $k(a)\in\naturals$ to denote
the $a$-component of $\mec k$.
As usual, $\naturals^n$ denotes $\naturals^{[n]}=\naturals^{\{1,\ldots,n\}}$.
We use similar conventions for $\naturals$ replaced by $\reals$,
$\complex$, etc.

If $A$ is a set, then $\# A$ denotes the cardinality of $A$.
We often denote a set with all capital letters, and its cardinality
in lower case letters; for example,
when we define
$\SNBC(G,k)$, we will write
$\snbc(G,k)$ for $\#\SNBC(G,k)$.

If $A'\subset A$ are sets, then $\II_{A'}\from A\to\{0,1\}$ (with $A$
understood) denotes
the characteristic function of $A'$, i.e., $\II_{A'}(a)$ is $1$ if
$a\in A'$ and otherwise is $0$;
we also write $\II_{A'}$ (with $A$ understood) to mean $\II_{A'\cap A}$
when $A'$ is not necessarily a subset of $A$.

All probability spaces are finite; hence a probability space
is a pair $\cP=(\Omega,P)$ where $\Omega$ is a finite set and
$P\from \Omega\to\reals_{>0}$ with $\sum_{\omega\in\Omega}P(\omega)=1$;
hence an {\em event} is any subset of $\Omega$.
We emphasize that $\omega\in\Omega$ implies that $P(\omega)>0$ with
strict inequality; we refer to the elements of $\Omega$ as
the atoms of the probability space.
We use $\cP$ and $\Omega$ interchangeably when $P$ is
understood and confusion is unlikely.

A {\em complex-valued random variable} on $\cP$ or $\Omega$
is a function $f\from\Omega\to\complex$, and similarly for real-,
integer-, and natural-valued random variable; we denote its
$\cP$-expected value by
$$
\EE_{\omega\in\Omega}[f(\omega)]=\sum_{\omega\in\Omega}f(\omega)P(\omega).
$$
If $\Omega'\subset\Omega$ we denote the probability of $\Omega'$ by
$$
\Prob_{\cP}[\Omega']=\sum_{\omega\in\Omega'}P(\omega')
=
\EE_{\omega\in\Omega}[\II_{\Omega'}(\omega)].
$$
At times we write $\Prob_{\cP}[\Omega']$ where $\Omega'$ is
not a subset of $\Omega$, by which we mean
$\Prob_{\cP}[\Omega'\cap\Omega]$.

\subsection{Graphs, Our Basic Models, Walks}

A {\em directed graph},
or simply a {\em digraph},
is a tuple $G=(V_G,\Edir_G,h_G,t_G)$ consisting of sets
$V_G$ and $\Edir_G$ (of {\em vertices} and {\em directed edges}) and maps
$h_G,t_G$ ({\em heads}
and {\em tails}) $\Edir_G\to V_G$.
Therefore our digraphs can have multiple edges and
self-loops (i.e., $e\in\Edir_G$ with $h_G(e)=t_G(e)$).
A {\em graph} is a tuple $G=(V_G,\Edir_G,h_G,t_G,\iota_G)$
where $(V_G,\Edir_G,h_G,t_G)$ is a digraph and
$\iota_G\from \Edir_G\to \Edir_G$ is an involution with
$t_G\iota_G=h_G$;
the {\em edge set} of $G$, denoted $E_G$, is the
set of orbits of $\iota_G$, which (notation aside)
can be identified with $\Edir_G/\iota_G$,
the set of equivalence classes of
$\Edir_G$ modulo $\iota_G$;
if $\{e\}\in E_G$ is a singleton, then necessarily $e$ is a self-loop
with $\iota_G e =e $, and
we call $e$ a {\em half-loop}; other elements of $E_G$ are sets
$\{e,\iota_G e\}$ of size two, i.e., with $e\ne\iota_G e$, and for such $e$
we say that $e$ (or, at times, $\{e,\iota_G e\}$)
is a {\em whole-loop} if
$h_G e=t_G e$ (otherwise $e$ has distinct endpoints).

Hence these definitions allow our graphs to have multiple edges and 
two types of self-loops---whole-loops
and half-loops---as in
\cite{friedman_geometric_aspects,friedman_alon}.
The {\em indegree} and {\em outdegree} of a vertex in a digraph is
the number of edges whose tail, respectively whose head, is the vertex;
the {\em degree} of a vertex in a graph is its indegree (which equals
its outdegree) in the underlying digraph; 
therefore a whole-loop about a vertex contributes $2$
to its degree, whereas a half-loop contributes $1$.

An {\em orientation} of a graph, $G$, is a choice $\Eor_G\subset\Edir_G$
of $\iota_G$ representatives; i.e., $\Eor_G$ contains every half-loop, $e$,
and one element of each two-element set $\{e,\iota_G e\}$.

A {\em morphism $\pi\from G\to H$} of directed graphs is a pair
$\pi=(\pi_V,\pi_E)$ where $\pi_V\from V_G\to V_H$ and
$\pi_E\from \Edir_G\to\Edir_H$ are maps that intertwine the heads maps
and the tails maps of $G,H$ in the evident fashion;
such a morphism is {\em covering} (respectively, {\em \'etale},
elsewhere called an {\em immersion}) if for each $v\in V_G$,
$\pi_E$ maps those directed edges whose head is $v$ bijectively
(respectively, injectively) to those whose head is $\pi_V(v)$,
and the same with tail replacing head.
If $G,H$ are graphs, then a morphism $\pi\from G\to H$ is a morphism
of underlying directed graphs where $\pi_E\iota_G=\iota_H\pi_E$;
$\pi$ is called {\em covering} or {\em \'etale} if it is so as a morphism
of underlying directed graphs.
We use the words {\em morphism} and {\em map} interchangeably.

A walk in a graph or digraph, $G$, is an alternating sequence
$w=(v_0,e_1,\ldots,e_k,v_k)$ of vertices and directed edges
with $t_Ge_i=v_{i-1}$ and $h_Ge_i=v_i$ for $i\in[k]$;
$w$ is {\em closed} if $v_k=v_0$;
if $G$ is a graph,
$w$ is {\em non-backtracking}, or simply {\em NB},
if $\iota_Ge_i\ne e_{i+1}$
for $i\in[k-1]$, and {\em strictly 
non-backtracking closed}, or simply {\em SNBC},
if it is closed, non-backtracking, and 
$\iota_G e_k\ne e_1$.
The {\em visited subgraph} of a walk, $w$, in a graph $G$, denoted
$\ViSu_G(w)$ or simply
$\ViSu(w)$, is the smallest subgraph of $G$ containing all the vertices
and directed edges of $w$;
$\ViSu_G(w)$ generally depends on $G$, i.e., $\ViSu_G(w)$ cannot be inferred
from the sequence $v_0,e_1,\ldots,e_k,v_k$ alone without knowing
$\iota_G$.

The adjacency matrix, $A_G$,
of a graph or digraph, $G$, is defined as usual (its $(v_1,v_2)$-entry
is the number of directed edges from $v_1$ to $v_2$);
if $G$ is a graph on $n$ vertices, 
then $A_G$ is symmetric and we order its eigenvalues (counted with
multiplicities) and denote them
$$
\lambda_1(G)\ge \cdots \ge \lambda_n(G).
$$
If $G$ is a graph, its
Hashimoto matrix (also called the non-backtracking matrix), $H_G$,
is the adjacency matrix of the {\em oriented line graph} of $G$,
$\Line(G)$,
whose vertices are $\Edir_G$ and whose directed edges
are the subset of $\Edir_G\times\Edir_G$ consisting of pairs $(e_1,e_2)$
such that $e_1,e_2$ form the
directed edges of a non-backtracking walk (of length two) in $G$
(the tail of $(e_1,e_2)$ is $e_1$, and its head $e_2$);
therefore $H_G$
is the square matrix indexed on $\Edir_G$, whose $(e_1,e_2)$ entry
is $1$ or $0$ according to, respectively, whether or not
$e_1,e_2$ form a non-backtracking walk
(i.e., $h_G e_1=t_G e_2$ and $\iota_G e_1\ne e_2$).
We use $\mu_1(G)$ to denote the Perron-Frobenius eigenvalue of 
$H_G$, and use $\mu_i(G)$ with $1<i\le \#\Edir_G$ to denote the
other eigenvalues of $H_G$ (which are generally complex-valued)
in any order.

If $B,G$ are both digraphs,
we say that $G$ is a {\em coordinatized graph over $B$
of degree $n$}
if
\begin{equation}\label{eq_coord_def}
V_G=V_B\times [n], \quad\Edir_G=\Edir_B\times[n], \quad
t_G(e,i)=(t_B e,i),\quad
h_G(e,i)=(h_Be,\sigma(e)i)
\end{equation} 
for some map
$\sigma\from\Edir_B\to\cS_n$, where $\cS_n$ is the group
of permutations on $[n]$; we call $\sigma$ (which is uniquely determined by
\eqref{eq_coord_def}) {\em the permutation assignment
associated to $G$}.
[Any such $G$ comes with a map $G\to B$ given by 
``projection to the first component of
the pair,'' and this map is a covering map of degree $n$.]
If $B,G$ are graphs, we say that a graph $G$ is a 
{\em coordinatized graph over $B$
of degree $n$} if \eqref{eq_coord_def} holds and also
\begin{equation}\label{eq_coord_def_graph}
\iota_G(e,i) = \bigl( \iota_B e,\sigma(e)i \bigr) ,
\end{equation} 
which implies that 
\begin{equation}\label{eq_sigma_iota_B}
(e,i)=\iota_G\iota_G(e,i) = \bigl( e, \sigma(\iota_B e)\sigma(e)i \bigr)
\quad\forall e\in\Edir_B,\ i\in[n],
\end{equation}
and hence $\sigma(\iota_B e)=\sigma(e)^{-1}$;
we use $\Coord_n(B)$ to denote the set of all coordinatized covers
of a graph, $B$, of degree $n$.

The {\em order} of a graph, $G$, is $\ord(G)\eqdef (\#E_G)-(\#V_G)$.
Note that a half-loop and a whole-loop each contribute $1$ to 
$\#E_G$ and to the order of $G$.
The {\em Euler characteristic} of a graph, $G$, is
$\chi(G)\eqdef (\# V_G) - (\#\Edir_G)/2$.
Hence $\ord(G)\ge -\chi(G)$, with equality iff $G$ has no half-loops.

If $w$ is a walk in any $G\in\Coord_n(B)$, then one easily
sees that $\ViSu_G(w)$ can be inferred
from $B$ and $w$ alone.

If $B$ is a graph without half-loops, then the {\em permutation model over
$B$} refers to the probability spaces $\{\cC_n(B)\}_{n\in\naturals}$ where
the atoms of $\cC_n(B)$ are coordinatized coverings of degree $n$
over $B$ chosen with the uniform distribution.
More generally, a {\em model} over a graph, $B$, is a collection of
probability spaces, $\{\cC_n(B)\}_{n\in N}$, 
defined for $n\in N$ where $N\subset\naturals$ is an
infinite subset, and where the atoms of each $\cC_n(B)$ are elements
of $\Coord_n(B)$.
There are a number of models related to the permutation model,
which are generalizations of the models of \cite{friedman_alon},
that we call {\em our basic models} and are defined in Article~I;
let us give a rough description.

All of {\em our basic models} are {\em edge independent}, meaning that
for any orientation $\Eor_B\subset\Edir_B$, the values of 
the permutation assignment, $\sigma$, on $\Eor_B$ are independent
of one another (of course, $\sigma(\iota_G e)=(\sigma(e))^{-1}$,
so $\sigma$ is determined by its values on any orientation
$\Eor_B$); for edge independent models, it suffices to specify
the ($\cS_n$-valued)
random variable $\sigma(e)$ for each $e$ in $\Eor_B$ or $\Edir_B$.
The permutation model can be alternatively described as the 
edge independent model that assigns a uniformly chosen permutation
to each $e\in\Edir_B$ (which requires $B$ to have no half-loops);
the {\em full cycle} (or simply {\em cyclic}) model is the same, except
that if $e$ is a whole-loop then $\sigma(e)$ is chosen uniformly
among all permutations whose cyclic structure consists of a single
$n$-cycle.
If $B$ has half-loops, then we restrict $\cC_n(B)$ either to $n$ even
or $n$ odd and for each half-loop $e\in\Edir_B$ we
choose $\sigma(e)$ as follows: if $n$ is even we choose 
$\sigma(e)$ uniformly among all perfect matchings,
i.e., involutions (maps equal to their inverse) with no fixed points;
if $n$ is odd then we choose $\sigma(e)$ uniformly among
all {\em nearly perfect matchings}, meaning involutions with one
fixed point.
We combine terms when $B$ has half-loops: for example,
the term {\em full cycle-involution} (or simply {\em cyclic-involution})
{\em model of odd degree over $B$} refers
to the model where the degree, $n$, is odd,
where $\sigma(e)$ follows the full cycle rule when $e$ is
not a half-loop, and where $\sigma(e)$ is a near perfect matching
when $e$ is a half-loop;
similarly for the {\em full cycle-involution} (or simply 
{\em cyclic-involution})
{\em model of even degree}
and the {\em permutation-involution model of even degree}
or {\em of odd degree}.

If $B$ is a graph, then a model, $\{\cC_n(B)\}_{n\in N}$, over $B$
may well have $N\ne \naturals$ (e.g., our basic models above when
$B$ has half-loops); in this case many formulas involving
the variable $n$ are only defined for $n\in N$.  For brevity, we
often do not explicitly write $n\in N$ in such formulas; 
for example we usually write
$$
\lim_{n\to\infty} \quad\mbox{to abbreviate}\quad
\lim_{n\in N,\ n\to\infty} \ .
$$
Also we often write simply $\cC_n(B)$ or $\{\cC_n(B)\}$ for
$\{\cC_n(B)\}_{n\in N}$ if confusion is unlikely to occur.

A graph is {\em pruned} if all its vertices are of degree at least
two (this differs from the more standard definition of {\em pruned} 
meaning that there are
no leaves).  If $w$ is any SNBC walk in a graph, $G$, then
we easily see that
$\ViSu_G(w)$ is necessarily pruned: i.e., any of its vertices must be
incident upon a whole-loop or two distinct edges
[note that a walk of length $k=1$ about a half-loop, $(v_0,e_1,v_1)$, by
definition, is not SNBC since $\iota_G e_k=e_1$].
It easily follows that $\ViSu_G(w)$ is contained in the graph
obtained from $G$ by repeatedly ``pruning any leaves''
(i.e., discarding any vertex of degree one and its incident edge)
from $G$.
Since our trace methods only concern (Hashimoto matrices and)
SNBC walks, it suffices to work with models $\cC_n(B)$ where
$B$ is pruned.
It is not hard to see that if $B$ is pruned and connected,
then $\ord(B)=0$ iff $B$ is a cycle,
and $\mu_1(B)>1$ iff $\chi(B)<0$;
this is formally proven in Article~III (Lemma~6.4).
Our theorems are not usually interesting unless $\mu_1(B)>\mu_1^{1/2}(B)$,
so we tend to restrict our main theorems
to the case $\mu_1(B)>1$ or, equivalently,
$\chi(B)<0$; some of our techniques work without these restrictions.

\subsection{Asymptotic Expansions}
\label{su_asymptotic_expansions}


A function $f\from\naturals\to\complex$ is a {\em polyexponential} if
it is a sum of functions $p(k)\mu^k$, where $p$ is a polynomial
and $\mu\in\complex$, with the convention
that for $\mu=0$ we understand $p(k)\mu^k$ to mean
any function that vanishes for sufficiently large $k$\footnote{
  This convention is used because then for any fixed matrix, $M$,
  any entry of $M^k$, as a function of $k$, is a polyexponential
  function of $k$; more specifically, the $\mu=0$ convention
  is due to the fact that a Jordan block of eigenvalue $0$ is
  nilpotent.
  }; we refer to the $\mu$
needed to express $f$ as the {\em exponents} or {\em bases} of $f$.
A function $f\from\naturals\to\complex$ is {\em of growth $\rho$}
for a $\rho\in\reals$ if $|f(k)|=o(1)(\rho+\epsilon)^k$ for any $\epsilon>0$.
A function $f\from\naturals\to\complex$ is $(B,\nu)$-bounded if it
is the sum of a function of growth $\nu$ plus a polyexponential function
whose bases are bounded by $\mu_1(B)$ (the Perron-Frobenius eigenvalue
of $H_B$); the {\em larger bases} of $f$ (with respect to $\nu$) are
those bases of the polyexponential function that are larger in
absolute value than $\nu$.
Moreover, such an $f$ is called {\em $(B,\nu)$-Ramanujan} if its
larger bases are all eigenvalues of $H_B$.

We say that a function $f=f(k,n)$ taking some subset of $\naturals^2$ to
$\complex$ has a 
{\em $(B,\nu)$-bounded expansion of order $r$} if for some
constant $C$ we have
\begin{equation}\label{eq_B_nu_defs_summ}
f(k,n) = c_0(k)+\cdots+c_{r-1}(k)+ O(1) c_r(k)/n^r,
\end{equation} 
whenever $f(k,n)$ is defined and $1\le k\le n^{1/2}/C$, where
for $0\le i\le r-1$, the $c_i(k)$ are $(B,\nu)$-bounded and $c_r(k)$
is of growth $\mu_1(B)$.
Furthermore, such an expansion is called {\em $(B,\nu)$-Ramanujan}
if for $0\le i\le r-1$, the $c_i(k)$ are {\em $(B,\nu)$-Ramanujan}.

Typically our functions $f(k,n)$ as in
\eqref{eq_B_nu_defs_summ} are defined for all $k\in\naturals$
and $n\in N$ for an infinite set $N\subset\naturals$ representing
the possible degrees of our random covering maps in the model
$\{\cC_n(B)\}_{n\in N}$ at hand.

\subsection{Tangles}
\label{su_tangles}

A {\em $(\ge\nu)$-tangle} is any 
connected graph, $\psi$, with $\mu_1(\psi)\ge\nu$,
where $\mu_1(\psi)$ denotes the Perron-Frobenius eigenvalue of $H_B$;
a {\em $(\ge\nu,<r)$-tangle} is any $(\ge\nu)$-tangle of order less than
$r$;
similarly for $(>\nu)$-tangles, i.e.,
$\psi$ satisfying the weak inequality $\mu_1(\psi)>\nu$,
and for $(>\nu,r)$-tangles.
We use ${\rm TangleFree}(\ge\nu,<r)$ to denote those graphs that don't
contain a subgraph that is $(\ge\nu,<r)$-tangle, and
${\rm HasTangles}(\ge\nu,<r)$ for those that do; we
never use $(>\nu)$-tangles in defining TangleFree and HasTangles,
for the technical reason
(see Article~III or Lemma~9.2 of \cite{friedman_alon}) that
for $\nu>1$ and any $r\in\naturals$
that there are only finitely many 
$(\ge\nu,<r)$-tangles, up to isomorphism, that are minimal
with respect to inclusion\footnote{
  By contrast, there are infinitely many minimal $(>\nu,<r)$-tangles
  for some values of $\nu>1$ and $r$: indeed, consider any connected pruned
  graph $\psi$, and set $r=\ord(\psi)+2$, $\nu=\mu_1(\psi)$.  Then if
  we fix two vertices in $\psi$ and let $\psi_s$ be the graph that is
  $\psi$ with an additional edge of length $s$ between these two 
  vertices, then $\psi_s$ is an $(>\nu,<r)$-tangle.  However, if
  $\psi'$ is $\psi$ with any single edge deleted, and $\psi'_s$ is 
  $\psi_s$ with this edge deleted, then one can show that
  $\mu_1(\psi'_s)<\nu$ for $s$ sufficiently large.  It follows that
  for $s$ sufficiently large, $\psi_s$ are minimal $(>\nu,<r)$-tangles.
}.

\subsection{$B$-Graphs, Ordered Graphs, and Strongly Algebraic Models}
\label{su_B_ordered_strongly_alg}

An {\em ordered graph}, $G^\og$, is a graph, $G$, endowed with an
{\em ordering}, meaning
an orientation (i.e., $\iota_G$-orbit representatives), 
$\Eor_G\subset\Edir_G$, 
and total orderings of $V_G$ and $E_G$;
a walk, $w=(v_0,\ldots,e_k,v_k)$ in a graph endows $\ViSu(w)$ with a
{\em first-encountered} ordering:
namely, $v\le v'$ if the first occurrence of $v$ comes before that
of $v'$ in the sequence $v_0,v_1,\ldots,v_k$,
similarly for $e\le e'$, and we orient each edge in the
order in which it is first traversed (some edges may be traversed
in only one direction).
We use $\ViSu^\og(w)$ to refer to $\ViSu(w)$ with this ordering.

A {\em morphism} $G^\og\to H^\og$ of ordered graphs is a morphism
$G\to H$ that respects the ordering in the evident fashion.
We are mostly interested in {\em isomorphisms} of ordered graphs;
we easily see that any isomorphism $G^\og\to G^\og$ must be the
identity morphism; it follows that if $G^\og$ and $H^\og$ are
isomorphic, then there is a unique isomorphism $G^\og\to H^\og$.

If $B$ is a graph, then a $B$-graph, $G_\Bg$, is a graph $G$ endowed 
with a map $G\to B$ (its {\em $B$-graph} structure).
A {\em morphism} $G_\Bg\to H_\Bg$ of $B$-graphs is a morphism
$G\to H$ that respects the $B$-structures in the evident sense.
An {\em ordered $B$-graph}, $G^\og_\Bg$, is a graph endowed with
both an ordering and a $B$-graph structure; a morphism of
ordered $B$-graphs is a morphism of the underlying graphs that
respects both the ordering and $B$-graph structures.
If $w$ is a walk in a $B$-graph, $G_\Bg$, we use $\ViSu_\Bg(w)$ to denote
$\ViSu(w)$ with the $B$-graph structure it inherits from $G$ in
the evident sense; we use $\ViSu_\Bg^\og(w)$ to denote
$\ViSu_\Bg(w)$ with its first-encountered ordering.

At times we drop the superscript $\,^\og$ and the subscript $\,_\Bg$;
for example, we write $G\in\Coord_n(B)$ instead of $G_\Bg\in\cC_n(B)$
(despite the fact that we constantly utilize
the $B$-graph structure on elements of
$\Coord_n(B)$).

A $B$-graph $G_\Bg$ is {\em covering} or {\'etale} if its structure
map $G\to B$ is.

If $\pi\from S\to B$ is a $B$-graph, we use
$\mec a=\mec a_{S_\Bg}$ to denote the vector
$\Edir_B\to\integers_{\ge 0}$ given by
$a_{S_\Bg}(e) = \# \pi^{-1}(e)$;
since $a_{S_\Bg}(\iota_B e) = a_{S_\Bg}(e)$ for all $e\in\Edir_B$,
we sometimes view $\mec a$ as a function $E_B\to\integers_{\ge 0}$, i.e.,
as the function taking $\{e,\iota_B e\}$ to 
$a_{S_\Bg}(e)=a_{S_\Bg}(\iota_B e)$.
We similarly define $\mec b_{S_\Bg}\from V_B\to\integers_{\ge 0}$ by
setting $b_{S_\Bg}(v) = \#\pi^{-1}(v)$.
If $w$ is a walk in a $B$-graph, we set $\mec a_w$ to be
$\mec a_{S_\Bg}$ where $S_\Bg=\ViSu_\Bg(w)$, and similarly for $\mec b_w$.
We refer to $\mec a,\mec b$ (in either context) as
{\em $B$-fibre counting functions}.

If $S_\Bg^\og$ is an ordered $B$-graph and $G_\Bg$ is a $B$-graph, we 
use $[S_\Bg^\og]\cap G_\Bg$ to denote the set of ordered graphs ${G'}_\Bg^\og$
such that $G'_\Bg\subset G_\Bg$ and ${G'}_\Bg^\og\isom S_\Bg^\og$
(as ordered $B$-graphs); this set is naturally identified with the
set of injective morphisms $S_\Bg\to G_\Bg$, and the cardinality of these
sets is independent of the ordering on $S_\Bg^\og$.


A $B$-graph, $S_\Bg$, or an ordered $B$-graph, $S_\Bg^\og$,
{\em occurs in a model $\{\cC_n(B)\}_{n\in N}$}
if for all sufficiently large
$n\in N$, $S_\Bg$ is isomorphic to a $B$-subgraph of some element
of $\cC_n(B)$; similary a graph, $S$, {\em occurs in 
$\{\cC_n(B)\}_{n\in N}$} if it can be endowed with a $B$-graph
structure, $S_\Bg$, that occurs in 
$\{\cC_n(B)\}_{n\in N}$.

A model $\{\cC_n(B)\}_{n\in N}$ of coverings of $B$ is {\em strongly
algebraic} if
\begin{enumerate}
\item for each $r\in\naturals$
there is a function, $g=g(k)$, of growth $\mu_1(B)$
such that if $k\le n/4$ we have
\begin{equation}\label{eq_algebraic_order_bound}
\EE_{G\in\cC_n(B)}[ \snbc_{\ge r}(G,k)] \le
g(k)/n^r
\end{equation}
where $\snbc_{\ge r}(G,k)$ is the number of SNBC walks of length
$k$ in $G$ whose visited subgraph is of order at least $r$;
\item
for any $r$ there exists
a function $g$ of growth $1$ and real $C>0$ such that the following
holds:
for any ordered $B$-graph, $S_\Bg^\og$, that is pruned and of
order less than $r$,
\begin{enumerate}
\item
if $S_\Bg$ occurs in $\cC_n(B)$, then for
$1\le \#\Edir_S\le n^{1/2}/C$,
\begin{equation}\label{eq_expansion_S}
\EE_{G\in\cC_n(B)}\Bigl[ \#\bigl([S_\Bg^\og]\cap G\bigr) \Bigr]
=
c_0 + \cdots + c_{r-1}/n^{r-1}
+ O(1) g(\# E_S) /n^r
\end{equation} 
where the $O(1)$ term is bounded in absolute value by $C$
(and therefore independent of $n$ and $S_\Bg$), and
where $c_i=c_i(S_\Bg)\in\reals$ such that
$c_i$ is $0$ if $i<\ord(S)$ and $c_i>0$ for $i=\ord(S)$;
and
\item
if $S_\Bg$ does not occur in $\cC_n(B)$, then for any
$n$ with $\#\Edir_S\le n^{1/2}/C$,
\begin{equation}\label{eq_zero_S_in_G}
\EE_{G\in\cC_n(B)}\Bigl[ \#\bigl([S_\Bg^\og]\cap G\bigr) \Bigr]
= 0 
\end{equation} 
(or, equivalently, no graph in $\cC_n(B)$ has a $B$-subgraph isomorphic to
$S_\Bg^\og$);
\end{enumerate}
\item
$c_0=c_0(S_\Bg)$ equals $1$ if $S$ is a cycle (i.e., $\ord(S)=0$ and
$S$ is connected) that occurs in $\cC_n(B)$;
\item
$S_\Bg$ occurs in $\cC_n(B)$ iff $S_\Bg$ is an \'etale $B$-graph
and $S$ has no half-loops; and
\item
there exist
polynomials $p_i=p_i(\mec a,\mec b)$ such that $p_0=1$
(i.e., identically 1), and for every
\'etale $B$-graph, $S_\Bg^\og$, we have that
\begin{equation}\label{eq_strongly_algebraic}
c_{\ord(S)+i}(S_\Bg) = p_i(\mec a_{S_\Bg},\mec b_{S_\Bg}) \ .
\end{equation}
\end{enumerate}
Notice that condition~(3), regarding $S$ that are cycles, is implied
by conditions~(4) and~(5); we leave in condition~(3) since this makes the
definition of {\em algebraic} (below) simpler.
Notice that \eqref{eq_expansion_S} and \eqref{eq_strongly_algebraic}
are the main reasons that we work with
ordered $B$-graphs: indeed, the coefficients depend only on
the $B$-fibre counting function $\mec a,\mec b$, which 
depend on the structure of
$S_\Bg^\og$ as a $B$-graph; this is not true if we don't work with
ordered graphs: i.e.,
\eqref{eq_expansion_S} fails to
hold if we replace $[S_\Bg^\og]$
with $[S_\Bg]$ (when $S_\Bg$ has nontrivial automorphisms), where
$[S_\Bg]\cap G$ refers to the number of $B$-subgraphs of $G$ isomorphic
to $S_\Bg$; the reason is that
$$
\#[S_\Bg^\og]\cap G_\Bg = \bigl( \#{\rm Aut}(S_\Bg)\bigr)
\bigl( \#[S_\Bg]\cap G_\Bg \bigr)
$$
where ${\rm Aut}(S_\Bg)$ is the group of automorphisms of $S_\Bg$, 
and it is $[S_\Bg^\og]\cap G_\Bg$ rather than $[S_\Bg]\cap G_\Bg$
that turns out to have the ``better'' properties;
see Section~6 of Article~I for examples.
Ordered graphs are convenient to use for a number of other reasons.

\ignore{
\myDeleteNote{Stuff deleted here and below on September 13, 2018.}
}

\subsection{Homotopy Type}

The homotopy type of a walk and of an ordered subgraph are defined
by {\em suppressing} its ``uninteresting'' vertices of degree two;
examples are given in Section~6 of Article~I.
Here is how we make this precise.

A {\em bead} in a graph is a vertex of degree two that is not
incident upon a self-loop.
Let $S$ be a graph and $V'\subset V_S$ be a {\em proper bead subset} of 
$V_S$,
meaning that $V'$ consists only of beads of $V$,
and that no connected component of $S$ has all its vertices in
$V'$ (this can only happen for connected components of $S$ that
are cycles);
we define the {\em bead suppression} $S/V'$ to be the following
graph: (1) its
vertex set $V_{S/V'}$
is $V''=V_S\setminus V'$, (2) its directed edges, $\Edir_{S/V'}$ consist
of
the {\em $V$'-beaded paths}, i.e., non-backtracking walks
in $S$ between elements of $V''$ whose intermediate vertices lie in $V'$,
(3) $t_{S/V'}$ and $h_{S/V'}$ give the first and last vertex of
the beaded path, and (4) $\iota_{S/V'}$ takes a beaded path
to its reverse walk
(i.e., takes $(v_0,e_1,\ldots,v_k)$ to
$(v_k,\iota_S e_k,\ldots,\iota_S e_1,v_0)$).
One can recover $S$ from the suppression $S/V'$ for pedantic reasons,
since we have defined its directed edges to be beaded paths of $S$.
If $S^\og=\ViSu^\og(w)$ where $w$ is a non-backtracking walk,
then the ordering of $S$ can be inferred by the naturally
corresponding order on $S/V'$, and we use $S^\og/V'$ to denote
$S/V'$ with this ordering.

Let $w$ be a non-backtracking walk in a graph, and 
$S^\og=\ViSu^\og(w)$ its visited
subgraph; the {\em reduction} of $w$ is the ordered graph,
$R^\og$, denoted $S^\og/V'$,
whose underlying graph is
$S/V'$ where $V'$ is the set of beads of $S$ except
the first and last vertices of $w$ (if one or both are beads),
and whose ordering is naturally arises from that on $S^\og$;
the {\em edge lengths} of $w$ is the function $E_{S/V'}\to\naturals$
taking an edge of $S/V'$ to the length of the beaded path it represents
in $S$;
we say that $w$ is {\em of homotopy type} $T^\og$ for any ordered
graph $T^\og$ that is isomorphic to $S^\og/V'$; in this case
the lengths of $S^\og/V'$ naturally give lengths $E_T\to\naturals$
by the unique isomorphism from $T^\og$ to $S^\og/V'$.
If $S^\og$ is the visited subgraph of a non-backtracking walk,
we define the reduction, homotopy type, and edge-lengths of $S^\og$ to
be that of the walk, since these notions depend only on $S^\og$ and
not the particular walk.

If $T$ is a graph and $\mec k\from E_T\to\naturals$ a function, then
we use $\VLG(T,\mec k)$ (for {\em variable-length graph}) to denote
any graph obtained from $T$ by gluing in a path of length $k(e)$
for each $e\in E_T$.  If $S^\og$ is of homotopy type $T^\og$
and $\mec k\from E_T\to \naturals$ its edge lengths,
then $\VLG(T,\mec k)$ is isomorphic to $S$ (as a graph).
Hence the construction of variable-length graphs is a sort of
inverse to bead suppression.

If $T^\og$ is an ordering on $T$ that arises as the first encountered
ordering of a non-backtracking walk on $T$ (whose visited subgraph
is all of $T$), then this ordering gives rise to a natural
ordering on $\VLG(T,\mec k)$ that we denote $\VLG^\og(T^\og,\mec k)$.
Again, this ordering on the variable-length graph is a sort of
inverse to bead suppression on ordered graphs.

\subsection{$B$-graphs and Wordings}

If $w_B=(v_0,e_1,\ldots,e_k,v_k)$ with $k\ge 1$ is a walk in a graph
$B$, then we can identify
$w_B$ with the string $e_1,e_2,\ldots,e_k$ over the alphabet
$\Edir_B$.
For technical reasons, the definitions below of
a {\em $B$-wording} and 
the {\em induced wording}, are given as strings over $\Edir_B$ rather
than the full alternating string of vertices and directed edges.
The reason is that 
doing this gives the correct notion of the {\em eigenvalues} of
an algebraic model (defined below).

Let $w$ be a non-backtracking walk in a $B$-graph, whose reduction
is $S^\og/V'$, and let
$S_\Bg^\og=\ViSu_\Bg^\og$.
Then the {\em wording induced by $w$} on $S^\og/V'$ is
the map $W$ from $\Edir_{S/V'}$ to strings in $\Edir_B$
of positive length, 
taking a
directed edge $e\in\Edir_{S/V'}$ to the string of $\Edir_B$ edges
in the non-backtracking walk in $B$
that lies under the walk in $S$ that it represents.
Abstractly, we say that a {\em $B$-wording} of a graph $T$
is a map $W$ from $\Edir_T$ to words over the alphabet
$\Edir_B$ that represent (the directed edges of)
non-backtracking walks in $B$ such that
(1) $W(\iota_T e)$ is the reverse word (corresponding to
the reverse walk) in $B$ of $W(e)$, 
(2) if $e\in\Edir_T$ is a half-loop, then $W(e)$ is of length one
whose single letter is a half-loop, and
(3) the tail of the first directed edge in $W(e)$ 
(corresponding to the first vertex in the associated walk in $B$)
depends only on $t_T e$;
the {\em edge-lengths} of $W$ is the function $E_T\to\naturals$
taking $e$ to the length of $W(e)$.
[Hence the wording induced by $w$ above is, indeed, a $B$-wording.]

Given a graph, $T$, and a $B$-wording $W$, there is a $B$-graph,
unique up to isomorphism, whose underlying graph is $\VLG(T,\mec k)$
where $\mec k$ is the edge-lengths of $W$, and where the $B$-graph
structure maps the non-backtracking walk in $\VLG(T,\mec k)$
corresponding to an $e\in\Edir_T$ to the non-backtracking walk in $B$
given by $W(e)$.
We denote any such $B$-graph by $\VLG(T,W)$; again this is
a sort of inverse to starting with a non-backtracking walk
and producing the wording it induces on its visited subgraph.

Notice that if $S_\Bg^\og=\VLG(T^\og,W)$ for a $B$-wording, $W$,
then the $B$-fibre counting functions
$\mec a_{S_\Bg}$ and $\mec b_{S_\Bg}$ can be
inferred from $W$, and we may therefore write $\mec a_W$ and
$\mec b_W$.

\subsection{Algebraic Models}

By a $B$-type we mean a pair $T^{\rm type}=(T,\cR)$ consisting
of a graph, $T$, and a map from $\Edir_T$ to the set
of regular languages over the alphabet $\Edir_B$ (in the sense of regular
language theory) such that
(1) all words in $\cR(e)$ are positive length strings corresponding to
non-backtracking walks in $B$, 
(2) if for $e\in\Edir_T$ we have $w=e_1\ldots e_k\in\cR(e)$,
then $w^R\eqdef \iota_B e_k\ldots\iota_B e_1$ lies in $\cR(\iota_T e)$,
and (3) if $W\from \Edir_T\to(\Edir_B)^*$ (where $(\Edir_B)^*$ is
the set of strings over $\Edir_B$) satisfies
$W(e)\in\cR(e)$ and $W(\iota_T e)=W(e)^R$ for all $e\in \Edir_T$,
then $W$ is a $B$-wording.
A $B$-wording $W$ of $T$ is {\em of type $T^{\rm type}$} if
$W(e)\in\cR(e)$ for each $e\in\Edir_T$.

Let $\cC_n(B)$ be a model that satisfies (1)--(3) of the definition
of strongly algebraic.
If $\cT$ a subset of $B$-graphs,
we say that the model is {\em algebraic restricted to $\cT$}
if 
either all $S_\Bg\in\cT$ occur in $\cC_n(B)$ or they all do not,
and if so
there are polynomials $p_0,p_1,\ldots$ such that
$c_i(S_\Bg)=p_i(S_\Bg)$ for any $S_\Bg\in\cT$. 
We say that $\cC_n(B)$ is {\em algebraic} if 
\begin{enumerate}
\item
setting $h(k)$ to be
the number of $B$-graph isomorphism classes of \'etale $B$-graphs
$S_\Bg$ such that $S$ is a cycle of length $k$ and $S$ does
not occur in $\cC_n(B)$, we have that 
$h$ is a function of growth $(d-1)^{1/2}$; and
\item
for any
pruned, ordered graph, $T^\og$, there is a finite number of
$B$-types, $T_j^{\rm type}=(T^\og,\cR_j)$, $j=1,\ldots,s$, 
such that (1) any $B$-wording, $W$, of $T$ belongs to exactly one
$\cR_j$, and
(2) $\cC_n(B)$ is algebraic when restricted to $T_j^{\rm type}$.
\end{enumerate}

[In Article~I we show that
if instead each $B$-wording belong to 
{\em at least one} $B$-type $T_j^{\rm type}$, then one can choose a
another set of
$B$-types that satisfy (2) and where each $B$-wording belongs
to {\em a unique} $B$-type;
however, the uniqueness
is ultimately needed in our proofs,
so we use uniqueness in our definition of algebraic.]

We remark that one can say that a walk, $w$, in a $B$-graph,
or an ordered $B$-graphs, $S_\Bg^\og$, is of {\em homotopy type $T^\og$},
but when $T$ has non-trivial automorphism one {\em cannot} say
that is of $B$-type $(T,\cR)$ unless---for example---one orders
$T$ and speaks of an {\em ordered $B$-type}, $(T^\og,\cR)$.
[This will be of concern only in Article~II.]

We define the {\em eigenvalues} of a regular language, $R$, to be the minimal
set $\mu_1,\ldots,\mu_m$ such that for any $k\ge 1$,
the number of words of length $k$ in the language
is given as
$$
\sum_{i=1}^m p_i(k)\mu_i^k
$$
for some polynomials $p_i=p_i(k)$, with the convention that
if $\mu_i=0$ then $p_i(k)\mu_i^k$ refers to any function that 
vanishes for $k$ sufficiently large (the reason for this is that
a Jordan block of eigenvalue $0$ is a nilpotent matrix).
Similarly, we define the eigenvalues of a $B$-type $T^{\rm type}=(T,\cR)$
as the union of all the eigenvalues of the $\cR(e)$.
Similarly a {\em set of eigenvalues} of a graph, $T$
(respectively, an algebraic model, $\cC_n(B)$)
is
any set containing the eigenvalues containing the eigenvalues
of some choice of $B$-types used in the definition of algebraic
for $T$-wordings (respectively, for $T$-wordings for all $T$).

[In Article~V we prove that all of our basic models are algebraic;
some of our basic models, such as the
permutation-involution model and the cyclic models, are not
strongly algebraic.]

We remark that a homotopy type, $T^\og$,
of a non-backtracking walk, can only have beads as its first or last 
vertices; however, in the definition of algebraic we require
a condition on {\em all pruned graphs}, $T$, 
which includes $T$ that may have many beads and may not be connected;
this is needed
when we define homotopy types of pairs in Article~II.

\subsection{SNBC Counting Functions}

If $T^\og$ is an ordered graph and $\mec k\from E_T\to\naturals$, 
we use $\SNBC(T^\og,\mec k;G,k)$ to denote the set of SNBC walks in $G$
of length $k$ and of homotopy type $T^\og$ and edge lengths $\mec k$.
We similarly define
$$
\SNBC(T^\og,\ge\bec\xi;G,k) \eqdef 
\bigcup_{\mec k\ge\bec\xi} \SNBC(T^\og,\mec k;G,k)
$$
where $\mec k\ge\bec\xi$ means that $k(e)\ge\xi(e)$ for all $e\in E_T$.
We denote the cardinality of these sets by replacing $\SNBC$ with
$\snbc$;
we call $\snbc(T^\og,\ge\bec\xi;G,k)$ the set of 
{\em $\bec\xi$-certified
traces of homotopy type $T^\og$ of length $k$ in $G$};
in Article~III we will refer to certain $\bec\xi$ as {\em certificates}.

\section{Main Theorems in This Article}
\label{se_main_theorems}

The main theorems in this article were stated in Article~I.
For ease of reading we restate them here.

Recall that if $A'\subset A$ are sets, then $\II_{A'}$ denotes the
indicator function of $A'$.

We also recall that 
a model, $\{\cC_n(B)\}_{n\in N}$, over a graph $B$
may well have $N\ne \naturals$ (e.g., our basic models above when
$B$ has half-loops); in this case many formulas involving
the variable $n$ are only defined for $n\in N$.  For brevity, we
often do not explicitly write $n\in N$ in such formulas; 
for example we usually write
$$
\lim_{n\to\infty} \quad\mbox{to abbreviate}\quad
\lim_{n\in N,\ n\to\infty} \ ;
$$
as another example, $(B,\nu)$-bounded expansions for a function
$f(k,n)$ only hold where $f$ is defined, which is typically for all 
$k\in\naturals$ but only $n\in N$.


\begin{theorem}\label{th_main_two_results}
Let $B$ be a connected graph with 
$\mu_1(B)>1$, and let 
$\{\cC_n(B)\}_{n\in N}$ be
an algebraic model over $B$.
Let $r>0$ be an integer and $\nu\ge\mu_1^{1/2}(B)$ be a real number.
Then 
\begin{equation}\label{eq_main_tech_result1}
f(k,n)\eqdef
\EE_{G\in\cC_n(B)}[ \II_{{\rm TangleFree}(\ge\nu,<r)}(G) \Trace(H^k_G) ]
\end{equation}
has a $(B,\nu)$-bounded expansion to order $r$,
$$
f(k,n)=c_0(k)+\cdots+c_{r-1}(k)/n^{r-1}+O(1)c_r(k)/n^r,
$$
where
\begin{equation}\label{eq_give_c_zero_k}
c_0(k)=\sum_{k'|k} \Trace(H_B^{k'}) 
\end{equation} 
where the sum is over all positive integers, $k'$, dividing $k$; hence
$$
c_0(k) = \Trace(H_B^k) + O(k) \mu_1^{k/2}(B);
$$
furthermore, the larger
bases of each $c_i(k)$ (with respect to $\mu_1^{1/2}(B)$)
is some subset of the eigenvalues
of the model.
Finally, for any $r'\in\naturals$ the function
\begin{equation}\label{eq_main_tech_result2}
\widetilde f(n)  \eqdef
\EE_{G\in\cC_n(B)}[ \II_{{\rm TangleFree}(\ge\nu,<r')}(G)]
=
\Prob_{\cC_n(B)}[{\rm TangleFree}(\ge\nu,<r') ]
\end{equation}
has an asymptotic expansion in $1/n$ to any order $r$,
$$
\widetilde c_0+\cdots+\widetilde c_{r-1}/n^{r-1}+O(1)/n^r ;
$$
where $\widetilde c_0=1$; furthermore, if $j_0$ is the 
smallest order of a $(\ge\nu)$-tangle occurring in $\cC_n(B)$,
then 
$c_j=0$ for $1\le j<j_0$ and
$c_j>0$ for $j=j_0$
(provided that $r\ge j_0+1$ so that
$\widetilde c_{j_0}$ is uniquely defined).
\end{theorem}

Notice that a model may have---at least in principle---an infinite number 
of eigenvalues, which
means that for each $r,\nu$, the number of bases of the
$c_i(k)$ may be unbounded as $i\to\infty$; 
however there are a few remarks to consider:
\begin{enumerate}
\item taking $\nu=\mu_1^{1/2}(B)$, for each $r$, the $c_i(k)$ with $i<r$
have a finite number of bases;
\item since for any fixed $k$ we have
$$
c_i(k)  = \lim_{n\to\infty} n^i \bigl( f(k,n) - c_0(k) - \cdots - 
n^{1-i}c_{i-1}(k) \bigr) ,
$$
the $c_i(k)$ is uniquely defined and
independent of $r>i$; hence a fixed $c_i(k)$ has a finite
number of larger (than $\mu_1^{1/2}(B)$) bases;
\item
in all our basic models, the eigenvalues of the model can be taken to be
the 
$\mu_j(B)$ and possibly the eigenvalue $1$; hence all larger bases
of any $c_i(k)$ lie in this finite set of eigenvalues.
\end{enumerate}

The other result that will be needed in later articles, namely Article~V,
is proven similarly to our proof of an expansion for the function
in \eqref{eq_main_tech_result2}.

\begin{theorem}\label{th_extra_needed}
Let $\cC_n(B)$ be an algebraic model over a graph, $B$, and
let $S_\Bg$ be a connected, pruned graph of positive order
that occurs in this model (recall that this
means that for some $n$ and some $G\in\cC_n(B)$, $G_\Bg$ has a subgraph
isomorphic to $S_\Bg$).  Then for some constant, $C'$, and
$n$ sufficiently large (and $n\in N$),
$$
\Prob_{G\in\cC_n(B)}\Bigl[ [S_\Bg]\cap G\ne\emptyset
\Bigr] \ge
C' n^{-\ord(S_\Bg)}.
$$
\end{theorem}

\section{Theorems From Article~II}
\label{se_art_II_theorems}

For ease of reading, let us recall the main theorems of Article~II,
which we will use here.

\begin{theorem}\label{th_main_certified_pairs}
Let $\cC_n(B)$ be an algebraic model over a graph $B$.
Let $T^\og$ be an ordered graph, let $\bec\xi\from E_T\to\naturals$ be
a function, and let
$$
\nu = \max\Bigl( \mu_1^{1/2}(B), \mu_1\bigl(\VLG(T,\bec\cert)\bigr) \Bigr).
$$
Let $\psi_\Bg^\og$ be any pruned, ordered $B$-graph.
Then for any
$r\ge 1$ we have
\begin{equation}\label{eq_subgraphs_times_walks}
f(k,n)=
\EE_{G\in\cC_n(B)}[ 
(\#[\psi_\Bg^\og]\cap G)
\snbc(T^\og;\ge\bec\xi,G,k) ]
\end{equation} 
has a $(B,\nu)$-bounded expansion 
$$
c_0(k)+\cdots+c_{r-1}(k)/n^{r-1}+ O(1) c_r(k)/n^r,
$$
to order $r$; the bases of
the coefficients in the expansion are some subset of the
eigenvalues of the model, and $c_i(k)=0$ for $i$ less than the
order of all $B$-graphs that contain both a walk of 
homotopy type $T^\og$
and a subgraph isomorphic to $\psi_\Bg$.
\end{theorem}

We note that the conclusions of this theorem also hold for the function
\begin{equation}\label{eq_psi_empty}
f(k,n)=\EE_{G\in\cC_n(B)}[
\snbc(T^\og;\ge\bec\xi,G,k) ] 
\end{equation}
(for the first part of this paper, dealing with expansions for the
expected certified traces, we use only need this particular
$f(k,n)$).
The reason is, as we now explain, in that the special case where
$\psi_\Bg^\og=\emptyset_\Bg^\og$ is the empty graph,
\eqref{eq_subgraphs_times_walks} reduces to
\eqref{eq_psi_empty};
the reader who dislikes the empty graph is free to simply
view the above theorem as also applying to \eqref{eq_psi_empty}
(this special case is stated both in Articles~I and~II), and with
$c_i(k)=0$ if $i<\ord(T)$.
The empty graph refers to the graph whose vertex and directed edge
sets are both the empty set, $\emptyset$; since there is a unique
map from $\emptyset$ to any set, there are unique heads, tails,
and edge involution, and a unique $B$-structure and ordering
for this graph.  Hence $[\emptyset_\Bg^\og]$ consists of this
single graph, and the $B$-graph $\emptyset_\Bg$ is a subgraph
of any $B$-graph; for this reason
$[\emptyset_\Bg^\og]\cap G_\Bg$ equals $\emptyset_\Bg^\og$ for
any $B$-graph, $G$, and hence $\#[\emptyset_\Bg^\og]\cap G_\Bg=1$
for all $G_\Bg$;
hence \eqref{eq_subgraphs_times_walks} reduces to
\eqref{eq_psi_empty} in this case.
Moreover, any $B$-graph contains the empty $B$-graph, and so
the condition on $i$ to have $c_i(k)=0$ amounts to $i$ being
less than the order of any $B$-graph of homotopy type $T^\og$,
which implies $i<\ord(T)$
(and is equivalent to $i<\ord(T)$ assuming at least one 
ordered $B$-graph of homotopy type $T^\og$ occurs in $\{\cC_n(B)\}$).

\section{Certified Traces and Theorem~\ref{th_main_two_results}}
\label{se_certified}

In this section we describe our approach to proving the main part of 
Theorem~\ref{th_main_two_results},
which is the existence of asymptotic expansions for
\eqref{eq_main_tech_result1} and facts about the coefficients $c_i=c_i(k)$;
the existence of asymptotic expansions for
\eqref{eq_main_tech_result2} follows easily from the facts we develop
for \eqref{eq_main_tech_result1}.

\subsection{Motivation for Modified Traces}

If $B$ is a bouquet of $d/2$ whole-loops (so $d$ is even), and
$\cC_n(B)$ is the permutation model, then \cite{friedman_alon} proves that
$$
f(k,n) \eqdef
\EE_{G\in\cC_n(B)}[\snbc(G,k)]
=\EE_{G\in\cC_n(B)}[\trace(H_G^k)]
$$ 
fails to have a 
$(B,\nu)$-Ramanujan expansion to all orders; the reason
is mainly due to the existence of
$(\ge\nu,<r)$-tangles that occur as subgraphs of graphs
in $\cC_n(B)$, where $\nu>\sqrt{d-1}$ and $r$ is of order $d^{1/2}$;
see the proof of Theorem~2.12 of
\cite{friedman_alon}.
Our remedy, as in \cite{friedman_alon},
will be to replace $\snbc(G,k)$ by a ``modification''
or ``regularization'' of this count, by counting
elements of $\SNBC(G,k)$ that satisfy a restrictive condition.
In \cite{friedman_alon}, these modified versions of $\snbc(G,k)$
were called {\em selective traces}; in this series of articles we use
the simpler {\em certified traces}.

\subsection{Definition of Certified Traces}

\begin{definition}\label{de_certified_traces}
Let $\nu>1$ be a real number, $r,k\in\integers_{\ge 0}$, and $G$ be a graph.
We define the set
of {\em $(<\nu,<r)$ (strictly) certified walks}
(respectively, {\em $(\le \nu,< r)$ (weakly) certified}),
denoted ${\rm CERT}_{<\nu,<r}(G,k)$ (respectively ${\rm CERT}_{\le\nu,<r}(G,k)$)
to be the set of SNBC walks in $G$
of length $k$ whose visited subgraph $S$ satisfies
$\mu_1(S)<\nu$ (respectively, $\mu_1(S)\le\nu$) and $\ord(S)<r$.
We define the {\em $(<\nu,<r)$ (weakly) certified trace}
(respectively, {\em $(\nu,<r)$ (strictly) certified trace}) 
{\em of $G$ of length $k$},
denoted ${\rm cert}_{<\nu,<r}(G,k)$
(respectively ${\rm cert}_{\le\nu,<r}(G,k)$)
to be the cardinality of ${\rm CERT}_{<\nu,<r}(G,k)$
(respectively, ${\rm CERT}_{\le\nu,<r}(G,k)$).
\end{definition}

The fundamental fact about certified traces is that 
\begin{equation}\label{eq_certified_versus_snbc}
G\in{\rm TangleFree}(\ge\nu,<r)\ \ \implies\quad
{\rm cert}_{<\nu',<r}(G,k) 
={\rm cert}_{\le\nu',<r}(G,k) 
= \snbc_{<r}(G,k)
\end{equation}
for any $r$ and $\nu'\ge \nu$.  
For this reason, the certified traces are sort of
``regularized'' SNBC count, that agrees with $\snbc_{<r}(G,k)$ for
$G$ without $\nu$-tangles of small order, but remains well controlled
for $G$ with such tangles.
The {\em selective traces} of \cite{friedman_alon} are another collection
of ``regularized traces,'' but are more cumbersome to define and utilize.
In this article we work with
strongly-certified traces ${\rm cert}_{<\nu,<r}(G,k)$,
although one could equally well work with
weakly-certified traces ${\rm cert}_{\le\nu,<r}(G,k)$.
All that our trace methods require is that we apply
\eqref{eq_certified_versus_snbc} with
$r\to\infty$ and $(d-1)^{1/2}<\nu\le\nu'\le (d-1)^{1/2}+\epsilon$
with $\epsilon\to 0$, and that we work with either
${\rm cert}_{<\nu',<r}(G,k)$ or
${\rm cert}_{\le\nu',<r}(G,k)$ there;
we prefer to take $\nu'=\nu$ for simplicity.
By contrast, we must work with 
${\rm TangleFree}(\ge\nu,<r)$,
rather than the analogously defined ${\rm TangleFree}(>\nu,<r)$, since
we need the number of (isomorphism classes of)
{\em minimal $(\ge\nu,<r)$-tangles} to be 
finite
(see the remarks concerning Lemma~\ref{le_finite_min_tangles}
in Section~\ref{se_finish_proof}).

\subsection{Statement of the Expansion Theorems for Certified Traces}

The main theorem in this paper is proven using the following two results.

\begin{theorem}\label{th_main_tech_result_for_cert}
Let $B$ be a connected graph with $\chi(B)<0$, and let 
$\{\cC_n(B)\}_{n\in N}$ be
an algebraic model over $B$.
Let $r'\in\naturals$ and $\nu\ge\mu_1^{1/2}(B)$ be a real number.
Then
\begin{equation}\label{eq_cert_exp_thm}
f(k,n)\eqdef
\EE_{G\in\cC_n(B)}[ {\rm cert}_{<\nu,<r'}(G,k) ]
\end{equation}
has a $(B,\nu)$-bounded asymptotic expansion to any order $r\in\naturals$,
$$
f(k,n)=c_0(k)+\cdots+c_{r-1}(k)/n^{r-1}+O(1)c_r(k)/n^r,
$$
where for some function, $h=h(k)$, of growth $(d-1)^{1/2}$ we have
\begin{equation}\label{eq_c_0}
c_0(k)=\sum_{k'|k} \Trace(H_B^{k'})  - h(k)
\end{equation} 
(the sum being over all positive integers, $k'$, dividing $k$),
and where the larger
bases of each $c_i(k)$ (with respect to $\mu_1^{1/2}(B)$)
is some subset of any set of eigenvalues
of the model.
Also, the function $h(k)$ in \eqref{eq_c_0} is precisely
the function described in condition~(1) of the definition
of {\em algebraic model}.
The same theorem holds if the (strictly-)certified trace in
\eqref{eq_cert_exp_thm} is replaced with the weakly-certified trace
${\rm cert}_{\le\nu,<r}(G,k)$.
\end{theorem}

Let ${\rm HasTangles}(\ge\nu,<r)$ denote the set of graphs, $G$, that
contain a $(\ge\nu,<r)$-tangle (as a subgraph); then 
${\rm HasTangles}(\ge\nu,<r)$ is the complement of 
${\rm TangleFree}(\ge\nu,<r)$, and so
$$
\II_{{\rm HasTangles}(\ge\nu,<r)}(G)
=
1- \II_{{\rm TangleFree}(\ge\nu,<r)}(G) \ .
$$

\begin{theorem}\label{th_main_tech_result_for_has_tangles}
Let $B$ be a connected graph with $\chi(B)<0$, and let $\{\cC_n(B)\}_{n\in N}$ 
be
an algebraic model over $B$.
Let $r,r',r''>0$ be integers and $\nu>1$ be a real number.
Then
\begin{equation}\label{eq_cert_has_tangles_exp_thm}
f(k,n)\eqdef
\EE_{G\in\cC_n(B)}[ 
\II_{{\rm HasTangles}(\ge\nu,<r'')}(G)
{\rm cert}_{<\nu,<r'}(G,k) ]
\end{equation}
has a $(B,\nu)$-bounded asymptotic expansion to order $r$, 
$$
f(k,n)=c_0(k)+\cdots+c_{r-1}(k)/n^{r-1}+O(1)c_r(k)/n^r,
$$
such that the
bases of the $c_i(k)$ are the eigenvalues
of the model; moreover,
$c_i$ vanishes if $i$ is less than the smallest
order of a $(\ge\nu,<r'')$-tangle that occurs in $\cC_n(B)$
provided that $i<r$
(i.e., occurs with positive probability in $\cC_n(B)$ for some $n$,
and hence for every $n$ sufficiently large).
In particular, $c_0(k)=0$ since $\nu>1$.
The same theorem holds if the strictly certified trace in
\eqref{eq_cert_exp_thm} is replaced with the weakly-certified trace
${\rm cert}_{\le\nu,<r'}(G,k)$.
\end{theorem}

Subtracting the above two results yields an expansion theorem to order $r$ for
$$
f(k,n)\eqdef
\EE_{G\in\cC_n(B)}[
\II_{{\rm TangleFree}(\ge\nu,<r)}(G)
{\rm cert}_{<\nu,<r}(G,k) ]
$$
with $c_0(k)$ given as in \eqref{eq_c_0} for $\nu>1$; in view of 
\eqref{eq_certified_versus_snbc}, this function is the same as
$$
f(k,n)\eqdef
\EE_{G\in\cC_n(B)}[
\II_{{\rm TangleFree}(\ge\nu,<r)}(G)
\snbc_{<r}(G,k) ] \ ;
$$
by \eqref{eq_algebraic_order_bound} we may replace $\snbc_{<r}(G,k)$
by $\snbc(G,k)$ at an additive cost bounded by
$C k^{2r}\Trace(H_B^k)/n^r$.
This proves the expansion theorem
in Theorem~\ref{th_main_two_results} for
\eqref{eq_main_tech_result1}.  The expansion theorem 
\eqref{eq_main_tech_result2} easily follows from the methods we use
to prove Theorem~\ref{th_main_tech_result_for_has_tangles}.

\subsection{Generalizations of Tangle Free Sets}

Theorems~\ref{th_main_tech_result_for_cert} and
\ref{th_main_tech_result_for_has_tangles} involve indicator function
for the set ${\rm TangleFree}(\ge\nu,<r)$.  In fact, we will prove
more general results where ${\rm TangleFree}(\ge\nu,<r)$ is replaced
with any set of graphs, $\cT$, subject to certain restrictions,
which we now describe.

\begin{definition}\label{de_positive}
We say that a graph is {\em pruned} if each of its vertices is of degree at
least two, and {\em positive} if it is pruned and moreover each of its
connected components is of positive order.
We say that a $B$-graph (or ordered graph, etc.) is positive if its
underlying graph is positive.
\end{definition}

\begin{definition}\label{de_meets_avoids}
Let $\cT$ be a class of graphs (respectively, of $B$-graphs, of ordered graphs,
etc.).
We say that a graph ($B$-graph, etc.)
$G$ {\em meets $\cT$} if $G$ has a non-empty subgraph ($B$-subgraph, etc.)
that is isomorphic to an element of $\cT$, and otherwise we say
$G$ {\em avoids $\cT$};
we use ${\rm Meets}(\cT)$ and ${\rm Avoids}(\cT)$ respectively
to be the class of graphs (or $B$-graphs, etc.) that meet and avoid
$\cT$.
\end{definition}
Although ${\rm TangleFree}(\ge\nu,<r)$ describes a class of graphs, for
various reasons we will want to work with the class of $B$-graphs
whose underlying graph lies in
${\rm TangleFree}(\ge\nu,<r)$; for this reason we make the above definition
for class of graphs and $B$-graphs.
The above notion of meeting and avoiding also makes sense for ordered graphs
and ordered $B$-graphs (and in other settings), but we will only be 
interested in graphs and $B$-graphs.

\begin{definition}\label{de_finitely_generated}
We say that a class of graphs ($B$-graphs, etc.)
$\cT$ is {\em finitely generated} if there is a finite
set $\cT'$ for which ${\rm Meets}(\cT)={\rm Meets}(\cT')$,
and {\em finitely positively generated} if there exists such a $\cT'$
such that each of its elements is positive.
\end{definition}
It is easy to see that if $\cT'\subset \cT$,
then ${\rm Meets}(\cT)={\rm Meets}(\cT')$ iff
$\cT'$ contains at least one graph in each isomorphism class of graphs
that are minimal with respect to inclusion (of graphs, of $B$-graphs, etc.)
in $\cT$.

Our interest in finitely positively generated classes is due to the following
proposition.
\begin{lemma}\label{le_finite_min_tangles}
For any real $\nu>1$ and $r\in\integers$,
${\rm TangleFree}(\ge\nu,<r)$ is finitely positively generated.
\end{lemma}
We will prove this in Section~\ref{se_finish_proof}, using the ideas 
of Lemma~9.2 of \cite{friedman_alon} and its proof there.

\subsection{Main Theorem for Indicator Functions}

\begin{theorem}\label{th_exp_ind}
Let $B$ be a graph, $\{\cC_n(B)\}_{n\in N}$ an algebraic model of $B$,
and let $\cT$ be a finitely positively generated class of graphs
or of $B$-graphs;
let $j$ be the smallest order of a graph in $\cT$ that occurs in
$\cC_n(B)$ (if $j$ doesn't exist we take $j=+\infty$).
Then the function
$$
f(n) \eqdef
\EE_{G\in\cC_n(B)}[ \II_{{\rm Meets}(\cT)}(G)]
=\Prob_{G\in\cC_n(B)}[{\rm Meets}(\cT)]
$$
has an asymptotic expansion in $1/n$ to any order $r$
$$
c_0 + c_1 n^{-1} + \cdots + c_{r-1} n^{-r+1} + O\bigl( n^{-r} \bigr)
$$
with $c_i=0$ if $i<j$ and, if $j\ne +\infty$, then $c_j>0$.
\end{theorem}
We remark that the case $j=+\infty$ in the above theorem is not
particularly interesting, since then $f(n)=0$ for all $n$.
Also, similar to a remark after Theorem~\ref{th_main_two_results}, the above
theorem implies that $c_i$ are given inductively as
$$
c_i =  \lim_{n\to\infty} 
\bigl( f(n) - c_0 + c_1 n^{-1} + \cdots + c_{i-1} n^{-i+1}\bigr)n^i,
$$
so that fact that $f(n)$ is a probability of some event implies that
$f(n)\in[0,1]$, and hence
$c_i\in\reals$ for all $i$, and the first nonzero $c_i$ must be positive.

\subsection{Main Theorems for Certified Traces with Indicator Functions}

\begin{theorem}\label{th_exp_ind_cert}
Let $\{\cC_n\}_{n\in N}$ be an algebraic model of random covering
maps over a graph $B$, with $\chi(B)<0$.
Let $r'\in\naturals$ and $\nu\ge \mu_1^{1/2}(B)$, and
let $\cT$ be a finitely positively generated class of graphs (or 
of $B$-graphs).
Then
$$
f(k,n)\eqdef
\EE_{G\in\cC_n(B)}[ \II_{{\rm Meets}(\cT)}(G){\rm cert}_{<\nu,<r'}(G;k) ]
$$
has a $(B,\nu)$-bounded asymptotic expansion to any order $r$,
$$
f(k,n)=c_0(k)+\cdots+c_{r-1}(k)/n^{r-1}+O(1)c_r(k)/n^r,
$$
where the bases of the $c_i$ 
are a subset of any set of eigenvalues of the model;
moreover $c_i(k)$ vanishes
for all $i$ less than the minimum order of a graph (or of a $B$-graph)
that contains
both some element of $\cT$ and a $(<\nu,<r')$-certified walk,
provided that $i<r$.
\end{theorem}

\subsection{Remarks on Theorems~\ref{th_exp_ind} 
and~\ref{th_exp_ind_cert} and Isomorphism Classes of $B$-Graphs}

In Sections~\ref{se_cert_proof}--\ref{se_ind_cert_proof}
we prove 
Theorems~\ref{th_exp_ind} and \ref{th_exp_ind_cert}, where $\cT$
is a finitely positively generated class of $B$-graphs (rather than graphs).
There is no harm in passing to $B$-graphs, since our models are
algebraic, and all our proof techniques work with $B$-graphs.
The case where $\cT$ is a set of graphs is equivalent to the
case where $\cT$ is replaced with the set of all $B$-graphs whose
underlying graph lies in $\cT$: if $\cT$ is a finitely generated
set of graphs, then the set of all possible $B$-graph structures
on the finite set of generators is finite.

We warn the reader of another change
in Sections~\ref{se_cert_proof}--\ref{se_ind_cert_proof}:
we work sets, $\Psi$, of {\em isomorphism classes} of $B$-graphs, rather
than $B$-graphs.
It is simpler to state 
Theorems~\ref{th_exp_ind} and \ref{th_exp_ind_cert} with $\Psi$ being
a set of $B$-graphs, which is why we have done so.
However, to prove these theorems we will work with
formulas---including those for 
{\em M\"obius functions}
and {\em indicator function approximations}---that are simpler
to define using
{\em isomorphism classes} of $B$-graphs.
Moreover,
some notions discussed already, such as being {\em finitely generated},
can be stated in terms of 
a finite number of
{\em isomorphism classes} of graphs (or of $B$-graphs).
In \cite{friedman_alon}, the symbol $\Psi$ with various subscripts
refers either to isomorphism classes of graphs, or to a set of representatives
in each isomorphism classes; in this article, we find it conceptually
simpler to give the proofs of the above theorems
using $\Psi$ to refer to a
set of {\em isomorphism classes} of $B$-graphs.

The small cost of working with isomorphism classes of $B$-graphs is
that one has to get used to slightly different terminology.
So in the next few sections we replace a class of graphs or $B$-graphs, 
$\cT$, with
a finite set of {\em isomorphism classes} of $B$-graphs
$$
\Psi = \{ [\psi_B^1], \ldots, [\psi_B^m] \}.
$$
One has to get used to speaking of $B$-graphs, $\psi_B$, lying in
{\em an element} of $\Psi$ (or of $\Psi^+$ or $\Psi^+_{<r}$, defined
in Section~\ref{se_indicator}), meaning $\psi_B\in [\psi_B^i]$ for some $i$
(rather than $\psi_B\in\cT$).
Hence we make the following definition.

\begin{definition}\label{de_meets_isom_classes}
Let $B$ be a graph and
$\Psi$ be a set of isomorphism classes of $B$-graphs.
We use ${\rm Meets}(\Psi)$ to denote the class of graphs, $G$, such that
some nonempty subgraph of $G$ is contained in an element $\Psi$.
\end{definition}

\section{Proof of Theorem~\ref{th_main_tech_result_for_cert}}
\label{se_cert_proof}

In this section we prove Theorem~\ref{th_main_tech_result_for_cert}.
The first two subsections each prove an easy preliminary lemma.

\subsection{VLG Comparisons}

\begin{lemma}\label{le_vlg_compare}
Let $T$ be a graph, and $\mec k,\mec k'$ be 
two maps $E_T\to\naturals$ with $\mec k\le \mec k'$
(i.e., $k(e)\le k'(e)$ for all $e\in E_T$).  Then
\begin{equation}\label{eq_vlg_compare}
\mu_1\bigl( \VLG(T,\mec k) \bigr)
\ge
\mu_1\bigl( \VLG(T,\mec k') \bigr) \ .
\end{equation} 
\end{lemma}

Its proof
is a standard consequence of ``Shannon's algorithm,''
and {\em majorization}
as described just above Theorem~3.5 of
\cite{friedman_alon}.
In the terminology there,
each entry of the matrix $Z_G(z)$, where $G$ is the oriented line graph of
$\VLG(T,\mec k)$, majorizes each of $Z_H(z)$ where $H$ is the oriented
line graph of
$\VLG(T,\mec k')$; hence each entry of $M_G(z)$ majorizes that of
$M_H(z)$; hence
equation~(12) and Theorem~3.5 of \cite{friedman_alon}
imply \eqref{eq_vlg_compare}.

[One can also prove \eqref{eq_vlg_compare} without Shannon's algorithm:
note that every
SNBC walk in $G=\VLG(T,\mec k)$ can be cyclically shifted by at most $\#E_G$
places to an SNBC walks in $G$ beginning at some vertex of $T$
(viewing $V_T$ as a subset of $V_G$); it follows that
$\mu_1\bigl( \VLG(T,\mec k) \bigr)$ is the limit as $m\to\infty$
of $W_m^{1/m}$, where $W_m$ is the number of SNBC walks beginning and
ending at a vertex of $T$ of length at most $m$.
But if $W'_m$ is the same quantity for
$G'=\VLG(T,\mec k')$, then we have $W_m\ge W'_m$
in view of the one-to-one correspondence of such SNBC
walks in $G'$ with those in $G$ (and with those in $T$),
for which the length of the walk in $G$ is at most the length of that in $G'$
since $\mec k\le\mec k'$.]

\subsection{The Finiteness of Minimal Elements in an Upper Subset of
$\naturals^n$}

The basis of our analysis of certified traces is the following 
finiteness lemma,
which we give after some definitions.

\begin{definition}
For an integer $n\ge 1$, endow $\naturals^n$ with the partial order
$\mec k\le \mec k'$ to mean that $k(i)\le k'(i)$ for all $i=1,\ldots,n$.
By an {\em upper set} in $\naturals^n$ we mean a subset, $U$, such
that if $\mec u\in U$ and $\mec u\le \mec u'$, then $\mec u'$.
\end{definition}

\begin{lemma}\label{le_finite_upper_mins}
Any upper set of $\naturals^n$ has a finite number of
minimal elements.
\end{lemma}
\begin{proof}
Let $U\subset \naturals^n$ be an upper set.  Let $x_1,\ldots,x_n$
be $n$ indeterminates, and let
$I\subset \complex[x_1,\ldots,x_n]$ be the set of polynomials
that are linear combinations of monomials
$$
\mec x^{\mec u} = x_1^{u_1}\ldots x_n^{u_n}
$$
with $\mec u\in U$.  Since $U$ is an upper set, $I$ is an ideal
of the ring $\complex[x_1,\ldots,x_n]$;
by Hilbert's Basis Theorem,
$I$ is finitely generated by polynomials, $p_1,\ldots,p_m\in I$.
For any $i=1,\ldots,m$ and
$q_i\in \complex[x_1,\ldots,x_n]$, each monomial $\mec x^{\mec w}$
appearing in
$p_iq_i$ arises as the product of some monomial in $p_i$ and some 
monomial in $q_i$.
It follows that any monomial $\mec x^{\mec w}$ that appears in a sum
$$
p_1 q_1 + \cdots p_m q_m
$$
has a corresponding monomial $\mec x^{\mec u}$ that appears in one of
$p_1,\ldots,p_m$, with $\mec u\le\mec w$.
But for any $\mec w\in U$,
since $\mec x^{\mec w}\in I$, we have
$$
\mec x^{\mec w} = p_1 q_1 + \cdots p_m q_m,
$$
for some $q_1,\ldots,q_m\in\complex[x_1,\ldots,x_n]$,
and hence there is a monomial $\mec x^{\mec u}$ appearing in $p_1,\ldots,p_m$
for which $\mec u\le\mec w$.
Hence the finite set of $\mec u$ such that $\mec x^{\mec u}$ appears in
one of $p_1,\ldots,p_m$, is a set of certificates for $U$.
\end{proof}
One can alternatively prove the above lemma directly: clearly
it holds for $n=1$; one can then prove the more general
lemma that if $\cP_1,\cP_2$ are
posets where any upper set has a finite number of minimal
elements, then the same is true of $\cP_1\times\cP_2$.

\subsection{The Zeroth Order Coefficient}

In this section we make the following observations about algebraic
models.  First we need a simple lemma.

\begin{lemma}\label{le_pruned_cycle_or_pos_ord}
Let $S$ be a connected and pruned graph.
Then either $S$ is a cycle, or $S$ is of positive order.
Also, $\mu_1(S)>1$ iff $\chi(S)<0$.
\end{lemma}
We will specifically need the first statement of the lemma in this
section.
The second statement in the above lemma is used in a number of
places in this series of articles,
to interchangeably use the conditions $\mu_1(S)>1$ and $\chi(S)<0$
for pruned, connected graphs.
Since both statements are based on a similar principle, we
prove both of them here.
\begin{proof}
To prove the first claim, note that
$$
\ord(S) = (1/2) \sum_{v\in V_S} \bigl( \deg'_S(v)-2 \bigr),
$$
where $\deg'$ is the usual degree of a vertex except that
each half-loop about a vertex, $v$, contributes
$2$ to its degree (instead of $1$ used for the usual degree and
Euler characteristic).
Since $S$ is pruned, we have $\deg'_S(v)\ge 2$ for each $v\in V_S$,
and hence $\ord(S)=0$ iff for $\deg'_S(v)=2$ for all $v\in V_S$.
It follows that each vertex of $V_S$
is either (1) incident upon two edges that 
are not self-loops, or (2) incident upon exactly one self-loop.
In case any vertex is incident upon a self-loop, then
the graph has one vertex and must be incident upon a whole-loop
(for otherwise the vertex would be of degree one);
hence $S$ is a cycle of length $1$.
Otherwise all vertices are of case (1), and therefore $S$ is a cycle.

To prove the second claim, we similarly note that
$$
\chi(S) = (1/2) \sum_{v\in V_S} \bigl( \deg_S(v)-2 \bigr).
$$
So if $S$ is pruned and $\chi(S)\le 0$, then 
$\deg_S(v)=2$ for
all $v$; it follows that $\chi(S)\le 0$ implies that
$S$ is either (1) a cycle, (2) a path
where each endpoint is incident upon an additional half-loop, 
(3) a single vertex incident upon a single whole-loop, 
or (4) a single vertex incident upon two half-loops.
In all these cases we easily check that $\mu_1(S)\le 1$,
since we easily see that there are at most two SNBC walks
of a given length about any vertex of $S$.

If $\chi(S)<0$, then some vertex of $S$ has degree $3$, say $v$.
Let us show that $\mu_1(S)>1$.

First we
claim that for any $e\in \Edir_S$ with $t_Se=v$, there is a
non-backtracking walk, $w$, about $v$ beginning with $e$: to see this,
we keep walking in a non-backtracking fashion, which we can do since
each vertex is of degree two, until we reach a repeated vertex; then 
return to $v$.  
Similarly, for any $e$ with $h_Se=v$, there is a non-backtracking
walk about $v$ ending in $e$.
So consider all pairs $(e,e')$ such that $t_Se=h_Se'=v$ and that
there exists a non-backtracking walk beginning in $e$ and ending in $e'$;
for each such pair, choose such a non-backtracking walk, $w_{e,e'}$; let
$m$ be an upper bound on the lengths of all these walks.
We now claim that for any $k$, the number of SNBC walks about $v$
of length at most $km$ is at least $2^{k-2}3$.
To see this, consider which walks of the form
$$
w_{e_1,e_2} w_{e_3,e_4} \ldots w_{e_{2k-1},e_{2k}}
$$
are SNBC: 
we may choose $e_1$ to be any of at least $3$ edges leaving $v$;
choosing some $e_2$ such that $w_{e_1,e_2}$ exists, we choose
$e_3$ to be any of at least $2$ edges leaving $v$ not equal to
$\iota e_2$; for $i=2,\ldots,k-2$ we similarly choose $e_{2i+1}$
to be an edge leaving $v$ not equal to $\iota e_{2i}$, of which there
are at least $2$ choices; finally
we choose $e_{2k-1}$ to be an edge leaving $v$ such that $w_{e_{2k-1},e_{2k}}$
exists and $e_{2k-1}\ne\iota e_{2k-2}$ and $e_{2k}\ne\iota e_1$,
of which there must be at least one choice.
Hence the trace of $H_S^{km}$ must be at least $2^{k-2}3$, and hence
$$
2^{k-2}3 \le \Trace(H_S^{km}) \le (\#\Edir_S)\mu_1^{km}(S);
$$
taking $k\to\infty$ we have
$\mu_1(S)\ge 2^{1/m}>1$.
\end{proof}

We easily see that no vertex of an SNBC walk can be of degree zero or one,
and hence we conclude the following corollary.

\begin{corollary}
Let $T^\og$ be the homotopy type of an SNBC walk.  Then either
$T^\og$ is the homotopy type of a cycle
(i.e., $T$ is the bouquet of a single whole-loop),
or else $\ord(T)\ge 1$.
\end{corollary}

\begin{corollary}\label{co_zeroth_order_coefs}
Let $B$ be a graph and ${\cC_n(B)}_{n\in N}$ an algebraic model over $B$.
Then there is a constant, $C$, and
a function $g$ of growth $\mu_1(B)$ such that
for $1\le k\le C/n^{1/2}$ we have
\begin{align*}
\EE_{G\in\cC_n(B)}[\snbc_0(G,k)] & = c_0(k) + g(k)O(1)/n, \\
\EE_{G\in\cC_n(B)}[\snbc_{\ge 1}(G,k)] & = g(k)O(1)/n ,
\end{align*}
where $c_0(k)$ is given in \eqref{eq_give_c_zero_k}.
\end{corollary}
\begin{proof}
The second equation follows from \eqref{eq_algebraic_order_bound}
(in the definition of strongly algebraic and algebraic).
If $S_\Bg^\og$ is a $B$-graph than is the visited subgraph of an SNBC
walk, $w$, of order $0$, then $S_\Bg^\og$ is necessarily of the homotopy
type of a cycle, and therefore the length $k'$, of $S$, must divide
$k$; furthermore, the directed edges over $B$ that lie over $w$ in
the first encountered ordering yield an SNBC walk in $B$ of length $k'$.
Conversely, every SNBC walk in $B$ of length $k'$ gives rise to an
ordered graph $S_\Bg^\og$, unique up to an isomorphism of $B$-graphs.
Since
each such $S_\Bg^\og$ has $c_0(S_\Bg)=1$, the first equation
of the corollary follows, since $\snbc(B,k')=\Trace(H_B^{k'})$.
\end{proof}

\subsection{Conclusion of The Proof of 
Theorem~\ref{th_main_tech_result_for_cert}}

\begin{proof}[Proof of Theorem~\ref{th_main_tech_result_for_cert}]
The elements of ${\rm CERT}_{<\nu,<r'}(G;k)$ are walks of a finite number
of homotopy types, $T^\og$.  
Hence it suffices to prove that for any fixed $T^\og$
$$
f(k,n) \eqdef 
\EE_{G\in\cC_n(B)}\Bigl[ 
\#\bigl( {\rm CERT}_{<\nu,<r'}(G,k) \cap \SNBC(T^\og;G,k) \bigr) 
\Bigr]
$$
has a $(B,\nu)$-asymptotic expansion to any order $r$, and that
its zero-th coefficient, $c_0(k)$ is given by
\begin{enumerate}
\item the formula \eqref{eq_give_c_zero_k} when
$T^\og$ is the homotopy type of a cycle (i.e., $T$ is the
bouquet of a single whole-loop), and
\item $c_0(k)=0$ otherwise.
\end{enumerate}

So fix an ordered graph, $T^\og$, and
consider the set
$$
U = \{ \mec k \from E_T\to\naturals \ | \ \VLG(T^\og,\mec k)<\nu \}  
$$
which is clearly an upper set by Lemma~\ref{le_vlg_compare}; according to
Lemma~\ref{le_finite_upper_mins}, this upper set has a finite number of
minimal elements
$$
\bec\xi^1,\ldots,\bec\xi^s \ .
$$

[Intuitively we think of each $\bec\xi_j$ as a {\em certificate} for
belonging in $U$, in that the condition $\mec k\ge\bec\xi_j$ certifies
(or guarantees)
that $\mec k\in U$.  The usefulness of the certified trace is due,
in part, to the fact that the condition $\VLG(T^\og,\mec k)<\nu$
is equivalent to being certified so by one of finitely many certificates.
This is why we use the name {\em certified trace}.]

For each $M\subset [s]$, we have
$$
\bigcap_{m\in M} \{ \mec k \ |\  \mec k\ge \bec\xi^m \} = 
\{ \mec k \ |\  \mec k\ge \bec\xi^M \}
$$
where
$$
\bec\xi^M \eqdef \max_{m\in M}(\bec\xi^m)
$$
is the component-wise maximum.
By inclusion/exclusion we have
\begin{equation}\label{eq_inclusion_exclusion_M}
\#\bigl( {\rm CERT}_{<\nu,<r'}(G,k) \cap \SNBC(T^{\og};G,k) \bigr)
=
\sum_{M\subset [s],\ M\ne\emptyset} (-1)^{1+(\#M)}
\snbc(T^{\og},\ge \bec\xi^M;G,k).
\end{equation} 
By Theorem~\ref{th_main_certified_pairs}, we have that for any $r>0$ and
any $M\subset [s]$ with $M\ne \emptyset$,
$$
f_M(k,n)  \eqdef
\EE_{G\in\cC_n(B)}[\snbc(T^{\og},\ge \bec\xi^M;G,k)]
$$
has a $(B,\nu_M)$-bounded expansion to order $r$ with
$$
\nu_M =  
\max\Bigl( \mu_1^{1/2}(B), 
\mu_1\bigl(\VLG(T,\bec\xi^M)\bigr) \Bigr).
$$
Since $U$ is an upper set, we have $\mec\xi^M\in U$ for all $M\ne \emptyset$,
and hence $\nu_M\le\nu$.
It follows that each $f_M(k,n)$ has an expansion that satisfies the
conditions in the statement of the theorem.  
Since $\cC_n(B)$ is a finite probability space,
we may take expected values in \eqref{eq_inclusion_exclusion_M} to conclude
that
$$
\EE_{G\in\cC_n(B)}\Bigl[
\#\bigl( {\rm CERT}_{<\nu,<r'}(G,k) \cap \SNBC(T^{\og};G,k) \bigr)
\Bigr]
=
\sum_{M\subset [s],\ M\ne\emptyset} (-1)^{1+(\#M)} f_M(k,n) \ ;
$$
since the RHS of this equation is a finite sum of functions with 
$(B,\nu)$-expansions to any order $r$,
so is each function
$$
f(k,n) = \EE_{G\in\cC_n(B)}
\Bigl[ 
\#\bigl( {\rm CERT}_{<\nu,<r'}(G,k) \cap \SNBC(T^{\og};G,k) \bigr)
\Bigr] \ .
$$
Summing over all the types, $T^{\og}$, of walks of order less than
$r'$, we conclude the same for
$$
f(k,n) = \EE_{G\in\cC_n(B)}
\bigl[ {\rm cert}_{<\nu,<r'}(G,k) \bigr],
$$
which proves the theorem for the expected value of the
strictly-certified trace in \eqref{eq_cert_exp_thm}.

It remains to compute $c_0(k)$.  For $T$ of order one or greater,
$$
\EE_{G\in\cC_n(B)}\Bigl[
\#\bigl( {\rm CERT}_{<\nu,<r'}(G,k) \cap \SNBC(T^{\og};G,k) \bigr)
\Bigr] \le
\EE_{G\in\cC_n(B)}\bigl[ \snbc_{\ge 1}(G,k) \bigr]
$$
which is bounded by $g(k)O(1)/n$ by Corollary~\ref{co_zeroth_order_coefs}.
On the other hand, any graph that is the homotopy type of a cycle
has $\mu_1=1<\nu$, and therefore if $T^\og$ is the homotopy type of a
cycle, we have
$$
\#\bigl( {\rm CERT}_{<\nu,<r'}(G,k) \cap \SNBC(T^{\og};G,k) \bigr)
=\snbc_0(G,k)
$$
for $k\ge 1$.  Hence for this $T$ we have
$$
\EE_{G\in\cC_n(B)}\Bigl[
\#\bigl( {\rm CERT}_{<\nu,<r'}(G,k) \cap \SNBC(T^{\og};G,k) \bigr)
\Bigr] =
\EE_{G\in\cC_n(B)}\bigl[ \snbc_0(G,k) \bigr]
$$
whose zeroth order coefficient, $c_0(k)$, is given
by Corollary~\ref{co_zeroth_order_coefs} to be
as in \eqref{eq_give_c_zero_k}.

This concludes the proof for strongly-certified traces.
The proof for weakly-certified traces is the same, with $U$ replaced by
$$
U' = \{ \mec k \from E_T\to\naturals \ | \ \VLG(T^\og,\mec k)\le\nu \} \ .
$$
This may change the set of certificates, i.e., of minimal elements,
$\bec\xi^1,\ldots,\bec\xi^s$, but everything else in the proof remains
the same.
\end{proof}

\section{Finiteness of Minimal Tangles}
\label{se_finiteness_min_tang}

We show that there are, up to isomorphism, only a finite number of minimal
$(\ge\nu,<r)$-tangles for any real $\nu>0$ and integer $r>0$.
This fact is not generally true of $(>\nu,<r)$-tangles.
This fact is essentially Lemma~9.2 of \cite{friedman_alon}, stated 
in the terms we use in this article.

\begin{definition}\label{de_minimal_tangle}
Let $\nu>0$ be a real number and $r>0$ an integer.  
We say that a graph, $\psi$, is
a {\em minimal} $(\ge\nu,<r)$-tangle if $\psi$ is a 
$(\ge\nu,<r)$-tangle, i.e.,
$\mu_1(\psi)\ge\nu$ and $\ord(\psi)<r$,  but all
of proper subgraphs of $\psi$ are not $(\ge\nu,<r)$-tangles.
\end{definition}

Let us recall Lemma~9.2 of \cite{friedman_alon} and its simple proof
based on the following ``continuity lemma'' regarding VLG's.

\begin{lemma}\label{le_shannon_cont}
Let $T$ be a fixed graph, and $E',E''$ partition of $E_T$ into
two sets.  Let $\mec k^1,\mec k^2,\ldots$ a sequence
of elements of $\naturals^{E_T}$ such that
\begin{enumerate}
\item $\mec k^i(e')$ is independent of $i$ for $e'\in E'$, and
\item $\mec k^i(e'')\to\infty$ as $i\to\infty$ for each $e''\in E''$.
\end{enumerate}
Let $T'$ be the graph obtained from $T$ by discarding $E''$ from $E_T$,
and let $\mec k'$ be the restriction of $\mec k^i$ to $E'$.
Say that there is a $\nu>1$ such that for all $i$ we have
\begin{equation}\label{eq_bounded_away_from_one}
\mu_1\bigl( \VLG(T,\mec k^i) \bigr) \ge \nu.
\end{equation} 
Then
$$
\lim_{i\to\infty}
\mu_1\bigl( \VLG(T,\mec k^i) \bigr)
=
\mu_1\bigl( \VLG(T',\mec k') \bigr) .
$$
\end{lemma}
We use 
{\em Shannon's algorithm}, as in
\cite{friedman_alon}, Theorem~3.6 and its proof;
we correct a minor error there: it is necessary that
\eqref{eq_bounded_away_from_one} hold with $\nu>1$, since if
$T$ has one vertex, one whole-loop, and
$E'=\emptyset$, then $\mu_1\bigl( \VLG(T,\mec k^i) \bigr) = 1$ for all
$i$, but $\mu_1\bigl( \VLG(T',\mec k') \bigr)=0$.
\begin{proof}
For any $\mec k\in\naturals^{E_T}$, 
let $M_{\mec k}(z)$ be the square matrix indexed on $\Edir_T$ whose $(e,e')$
entry is $0$ if the corresponding of $H_T$ is zero (i.e., $e,e'$ are
not the directed edges of a non-backtracking walk of length $2$),
and otherwise this entry is $z^{k(e)}$.  Let us prove that
$\mu_1\bigl( \VLG(T,\mec k) \bigr)$ is the reciprocal of the smallest
positive root $z=z_{\mec k}$ of the polynomial equation
\begin{equation}\label{eq_M_z_det_zero}
\det\bigl( I - M(z) \bigr) = 0 , \quad M(z)=M_{\mec k}(z) .
\end{equation} 
To see this,
let $G=\VLG(T,\mec k)$, and let $G'=\Line(G)$ be the oriented line graph
of $G$; hence $V_{G'}=\Edir_G$, and for 
$e,e'\in V_{G'}=\Edir_G$ there is one or zero edges from
$e$ to $e'$ according to whether or not $e,e'$ are the directed edges
of a non-backtracking walk of length two.
For each directed edge of $T$, $e_T\in\Edir_T$, let $\tilde e_T\in V_{G'}$
denote the first directed edge in the (beaded) directed walk path in $G'$
corresponding to $e_T$; let $\tilde E$ be the union of all the $\tilde e_T$
with $e_T\in\Edir_T$.
Then we easily see that
$G'$ is the same graph as the variable-length graph on its subset
of vertices, $\tilde E$, where two vertices, $\tilde e_T,\tilde e_T'$
have either one or zero edges from $\tilde e_T$ to $\tilde e_T'$
iff they form a non-backtracking walk of length two in $\Edir_T$,
and, if so, the length of the edge is $k(e_T)$.
Hence $\mu_1(G)$, which equals the Perron-Frobenius eigenvalue 
of $G'$, is given by \eqref{eq_M_z_det_zero}
in view of 
{\em Shannon's algorithm} (e.g., Theorem~3.5 of \cite{friedman_alon}, but
see the much earlier references in Section~3.2 of \cite{friedman_alon}).

Next say that an element of $V_{\Line(T)}=\Edir_T$ {\em belongs to $E'$}
if its $\iota_T$ orbit, i.e., its corresponding edge, lies in $E'$,
and similarly for $E''$.  Then we may partition $\Edir_T$ into two
sets: those belonging to $E'$ and those belonging to $E''$, and this
gives a block representation of $M(z)$ as
$$
M_{\mec k^i}(z) = 
\begin{bmatrix} M_1(z) & M_{3,i}(z) \\ M_2(z) & M_{4,i}(z) 
\end{bmatrix},
$$
where $M_1(z),M_2(z)$ are blocks that are independent of $i$,
since the edge-lengths $\mec k^i$ are constant on directed edges belonging
to $E'$, and each non-zero entry of $M_{3,i}(z),M_{4,i}(z)$ is a power
of $z$ than tends to infinity as $i\to\infty$.
In view of the last paragraph,
we have that $\mu_1( \VLG(T,\mec k^i))$ is the
reciprocal of the smallest positive root $z_i$ to
$$
\det\bigl( I - M_{\mec k^i}(z) \bigr) = 0.
$$
On the other hand, the above paragraph also shows that 
$\mu_1(\VLG(T',\mec k'))$ is the reciprocal of the smallest positive root,
$z$, of
\begin{equation}\label{eq_M_one_det}
\det\bigl( I' - M_1(z) \bigr) = 0
\end{equation} 
where $I'$ is the identity matrix indexed on directed edges belonging
to $E'$; by convention, we allow $z=+\infty$ if 
\eqref{eq_M_one_det} has no positive roots, in which case
$\mu_1(\VLG(T',\mec k'))=0$.  It remains to prove that
\begin{equation}\label{eq_limit_equals_z_0}
\lim_{i\to\infty} z_i = z_0.
\end{equation} 

Let us first show that the above limit exists.
Since $z_i< 1/\nu$, the sequence $\{z_i\}$ is bounded above; let
$z_\infty $ be its least upper bound; clearly $z_\infty\le 1/\nu<1$.
By definition, there exist $i_1,i_2,\ldots$ such that
$z_{i_n}\to z_\infty$ as $n\to\infty$.
Let $i=i_a$ for a fixed $a\in\naturals$.
We have $\mec k^i\le \mec k^j$ for $j$ sufficiently large, and hence
then Lemma~\ref{le_vlg_compare} implies that for such $j$,
$1/z_i>1/z_j$, i.e., $z_i<z_j$.  Hence
$$
\liminf_{n\to\infty} z_j \ge z_i=z_{i_a}.
$$
Taking $a\to\infty$ shows that
$$
\liminf_{j\to\infty} z_j \ge z_\infty = \limsup_{n\to\infty} z_i.
$$
Hence the limit in \eqref{eq_limit_equals_z_0} exists.  
To prove the lemma it suffices to show that
$$
z_\infty \le z_0 \quad\mbox{and}\quad z_\infty \ge z_0.
$$

Since $\VLG(T',\mec k')$ is a subgraph of each $\VLG(T,\mec k^i)$, 
we have
$$
1/z_0 = \mu_1\bigl( \VLG(T',\mec k') \bigr) 
\le
\mu_1\bigl( \VLG(T,\mec k^i) \bigr)  = 1/z_i ,
$$
and hence, taking $i\to\infty$, $1/z_0 \le 1/z_\infty$.
Hence $z_\infty\le z_0$.
Now let us show $z_0 \le z_\infty$.

Since $z_i\le z_\infty\le 1/\nu < 1$ for each $i$, we have that 
$M_{3,i}(z),M_{4,i}(z)$ tend to zero as $i\to\infty$, and hence
$$
M_{\mec k^i}(z_i) =
\begin{bmatrix} M_1(z_i) & M_{3,i}(z_i) \\ M_2(z) & M_{4,i}(z_i)
\end{bmatrix}
\to
\begin{bmatrix} M_1(z_\infty) & 0 \\ M_2(z_\infty) & 0
\end{bmatrix}
$$
as $i\to\infty$; then \eqref{eq_M_z_det_zero} (and the continuity
of the determinant) implies that
$$
\det\left( 
\begin{bmatrix} I' & 0 \\ 0 & I''
\end{bmatrix}
-
\begin{bmatrix} M_1(z_\infty) & 0 \\ M_2(z_\infty) & 0
\end{bmatrix} \right)
= 0,
$$
where $I''$ is the block identity matrix corresponding to directed edges
belonging to $E''$.  Hence
$$
\det\left( 
\begin{bmatrix} I'- M_1(z_\infty) & 0 \\ - M_2(z_\infty) & I''
\end{bmatrix} \right) 
= 0
$$
and therefore
$$
\det \bigl( I'- M_1(z_\infty) \bigr) 
=
\det\left(
\begin{bmatrix} I'- M_1(z_\infty) & 0 \\ - M_2(z_\infty) & I''
\end{bmatrix} \right)
= 0.
$$
Hence $z_\infty$ is a positive root of 
\eqref{eq_M_one_det}.  Therefore $z_0\le z_\infty$.
\end{proof}

\begin{theorem}\label{th_tangle_finite}
Let $\nu>1$ be a real number and $r\ge 0$ an integer.  The number
of (isomorphism classes of)
minimal $(\ge\nu,<r)$-tangles is finite.
\end{theorem}
\begin{proof}
Assume, to the contrary, that the theorem is false for some $r$; since
there are only finitely many homotopy types of order less than $r$,
then must exist an ordered graph,
$T^\og$,
and an infinite sequence of distinct minimal $\nu$-tangles
$\psi^1,\psi^2,\ldots$ of homotopy type $T^\og$ such that
$\mu_1(\psi^i)\ge \nu$.  Then we have $\psi^i=\VLG(T;\mec k^i)$
for a sequence of distinct vectors $\mec k^i\in\naturals^{E_T}$.
Let $E_T=\{e_1,\ldots,e_m\}$;
by passing to a subsequence of the $\mec k^i$, 
we may assume that 
either (1) $k^i(e_1)\to\infty$ as $i\to\infty$,
or that (2) the $k^i(e_1)$ are bounded, and hence---by passing to a
further subsequence---that $k^i(e_1)$ is independent of $i$;
we then repeat this process to show that for $j=2,\ldots,m$,
by passing to a subsequence we may assume that $k^i(e_j)\to\infty$
as $i\to\infty$ or $k^i(e_j)$ is independent of $i$.
At this point the $\mec k^i$ satisfy the
hypotheses of Lemma~\ref{le_shannon_cont}.
It follows that
$\VLG(T',\mec k')$ is a proper subgraph of $\VLG(T,\mec k^i)\isom\psi^i$
for all $i$, and $\mu_1(\VLG(T',k'))\ge \nu$,
contradicting the minimality of the $\psi^i$.
\end{proof}

We remark that the above lemma would be false for 
$(>\nu,<r)$-tangles, defined the same but with strict inequality between
$\mu_1$ and $\nu$, as was explained in a footnote in
Subsection~\ref{su_tangles}.

\section{Indicator Function Approximation}
\label{se_indicator}

In this section we develop some foundations regarding approximations
$I_r(\Psi,G)$
we shall use for 
the indicator function 
$$
\II_{{\rm Meets}(\Psi)}(G_\Bg)
$$
where $\Psi$ is a collection of $B$-graphs that satisfy slighter
milder assumptions
than those of Theorems~\ref{th_exp_ind} and \ref{th_exp_ind_cert}.
The results in this section are
adaptations of results of Section~9 of \cite{friedman_alon}.
In the first subsection we will state all the definitions and results
we will use in Sections~\ref{se_ind_cert_proof} and \ref{se_ind_proof};
the remaining subsections are devoted to their proofs.

\subsection{The Main Results}

\begin{definition}\label{de_derived_graphs}
Let $B$ be a graph, and let
$$
\Psi = \{ [\psi_\Bg^1],\ldots,[\psi_\Bg^m] \}
$$
be a finite set of isomorphism classes of $B$-graphs.
By {\em the set
derived $B$-graphs of $\Psi$}, denoted \( \Psi^+ \),
we mean the isomorphism classes $[\psi_\Bg]$ of $B$-graphs 
such that
$\psi$ can be written as the union of $B$-subgraphs each isomorphic
to some $\psi_\Bg^i$.
We use
$\Psi^+_{<r}$ to denote classes $[\psi_\Bg]$ of $\Psi^+$ with
$\ord(\psi)<r$.
If $G_\Bg$ is any $B$-graph, the {\em $\Psi$-image in $G$},
denoted, $\Psi^+\cap G_\Bg$, is the union of
all $B$-subgraphs of $G_\Bg$ that lie in some class, $[\psi_\Bg^i]$, of
$\Psi$ (clearly $\Psi^+\cap G_\Bg$ is largest subgraph of $G_\Bg$ that lies
in a class of $\Psi^+$); we use $\ord_\Psi(G)$ to denote
$\ord(\Psi^+\cap G_\Bg)$.
\end{definition}

\begin{definition}\label{de_injective_morphisms}
Let $B$ be a graph and
$\psi_\Bg,G_\Bg$ be two $B$-graphs.
By an {\em injective morphism} $\psi_\Bg\to G_\Bg$ we mean a morphism
that is injective as a map of vertex sets and of edge sets.
We use $N(\psi_\Bg,G_\Bg)$ to denote the number of injective maps
$\psi_\Bg\to G_\Bg$.
\end{definition}

The proposition below is worth stating, but easy to prove.
\begin{proposition}
For any graph $B$ and any two $B$-graphs $\psi_\Bg,G_\Bg$ we have that
for any ordering $\psi_\Bg^\og$ on $\psi_\Bg$,
$$
N(\psi_\Bg,G_\Bg) = \# ([\psi^\og_\Bg]\cap G_\Bg).
$$
\end{proposition}
\begin{proof}
Fix an ordering $\psi^\og_\Bg$; 
each injective
morphism $u\from\psi_\Bg\to G_\Bg$ determines an element
$S^\og_\Bg=u(\psi_\Bg^\og)$ of $[\psi^\og_\Bg]\cap G_\Bg$;
furthermore, for
$S^\og_\Bg\in [\psi^\og_\Bg]\cap G_\Bg$ there is a unique
isomorphism $\psi^\og_\Bg\to S^\og_\Bg$ giving rise to
an injection $\psi_\Bg\to G_\Bg$.
We easily check that this correspondence between injective maps $u$
and elements of $[\psi^\og_\Bg]\cap G_\Bg$ are inverses
of each other.
\end{proof}

We now state three results that will be proven
in the subsections that follow this one.

\begin{lemma}\label{le_finiteness_and_positivity}
Let $B$ be a graph, and let
$$
\Psi = \{ [\psi_\Bg^1],\ldots,[\psi_\Bg^m] \}
$$
be a finite set of isomorphism classes of $B$-graphs.
If each $\psi^i$ is positive, then for each $r\in\naturals$,
$\Psi^+_{<r}$ is finite.
\end{lemma}

\begin{lemma}\label{le_morphisms_factor}
Let $B$ be a graph, and let
$$
\Psi = \{ [\psi_\Bg^1],\ldots,[\psi_\Bg^m] \}
$$
be a finite set of isomorphism classes of $B$-graphs such $\Psi^+_{<r}$
is finite for all $r\in\naturals$; let $s$ be the largest number of
edges among the graphs $\psi^1,\ldots,\psi^m$.
\begin{enumerate}
\item 
If $[\psi_\Bg]\in\Psi^+$ and $G_\Bg$ is any $B$-graph, then
any injective morphism $\psi_\Bg\to G_\Bg$ factors as an injective
map $\psi_\Bg\to \Psi^+\cap G_\Bg$ followed by the inclusion of
$\Psi^+\cap G_\Bg$ in $G_\Bg$.
\item
If $r\in\naturals$
and $G_\Bg$ is a $B$-graph with $\ord_{\Psi}(G)\ge r$, 
and $[\psi_\Bg]\in\Psi^+_{<r}$, then
any injective map
$\psi_\Bg\to G_\Bg$ factors as two injective maps 
$\psi_\Bg\to\psi'_\Bg$ and $\psi'_\Bg\to G_\Bg$ where
$\psi'_\Bg$ lies in an element of $\Psi^+_{<r+s}\setminus \Psi^+_{<r}$,
i.e., $[\psi'_\Bg]$ 
is an element of $\Psi^+$ whose order is between $r$ and $r+s-1$.
In particular
\begin{equation}\label{eq_factor_inequality}
N(\psi_\Bg,G_\Bg) \le \sum_{\psi'_\Bg\in \Psi^+_{<r+s}\setminus \Psi^+_{<r}}
N(\psi_\Bg,\psi'_\Bg)
N(\psi'_\Bg,B_\Bg) .
\end{equation} 
\end{enumerate}
\end{lemma}

\begin{theorem}\label{th_mobius_inversion_truncation}
Let $B$ be a graph, and let
$$
\Psi = \{ [\psi_\Bg^1],\ldots,[\psi_\Bg^m] \}
$$
be a finite set of isomorphism classes of $B$-graphs such $\Psi^+_{<s}$
is finite for all $s\in\naturals$ and each $\psi^i$ is pruned.  
Then for each $[\psi_\Bg]\in\Psi^+$
there is a rational number $\mu[\psi_\Bg]$ such that
for any $r\in\naturals$ 
the function
$$
I_r(\Psi,G_\Bg) = \sum_{[\psi_\Bg]\in\Psi^+_{<r}} \#([\psi_\Bg]\cap G)
\mu[\psi_\Bg],
$$
satisfies
\begin{enumerate}
\item 
\begin{equation}\label{eq_trunc_exact}
I_r(\Psi,G_\Bg) = \II_{{\rm Meets}(\Psi)}(G_\Bg)
\end{equation}
whenever $\ord_\Psi(G_\Bg)<r$, and
\item
there is a constant $C\in\reals$ and $s\in\naturals$ such that if 
$\ord_\Psi(G_\Bg)\ge r$, then
\begin{equation}\label{eq_reasonable_trunc_bound}
\bigl| I_r(\Psi,G_\Bg) \bigr| \le C 
\sum_{[\psi_\Bg]\in \Psi^+_{<r+s}\setminus\Psi^+_{<r}} \#([\psi_\Bg]\cap G).
\end{equation} 
\end{enumerate}
\end{theorem}

\subsection{Proof of Lemma~\ref{le_finiteness_and_positivity}}

Lemma~\ref{le_finiteness_and_positivity} follows from
the proof Lemma~9.2 of \cite{friedman_alon}, which proves the same
in the context of graphs as opposed to
$B$-graphs.
For convenience we provide a complete proof here.
Our proof will use the following graph theoretic lemma.
\begin{lemma}\label{le_pruned_inclusions}
Let $G_1\subset G$ be pruned graphs, i.e., no vertex is isolated or
is incident upon a single edge that is not a self-loop.
Then if $G_1\ne G$, $\ord(G_1)<\ord(G)$.
\end{lemma}
This is Lemma~4.10 of \cite{friedman_alon}; for ease of reading we provide
a proof here.
\begin{proof}
Assume that $G_1\ne G$.
If $V_{G_1}=V_G$, then there is some edge in $G$ that is not in $G_1$,
and hence $\ord(G_1)<\ord(G)$.

Otherwise
$V_{G_1}\ne V_G$.  Since
$G$ is connected, there is at least one vertex $v_1\in V_G\setminus V_{G_1}$
of distance one to $V_{G_1}$; let $e_1$ be an edge connecting $v_1$ to a
vertex in $V_{G_1}$.  Continuing in this fashion, we get edges
$e_2,\ldots,e_t$ and vertices $v_2,\ldots,v_t$
such that the graph $G'$, obtained as the union
of $G_1$ and $e_1,\ldots,e_t$ and $v_1,\ldots,v_t$ has $t$ more vertices
than $V_{G_1}$, $t$ more edges (that are not self-loops), and
$V_{G'}=V_G$.
Note that $\ord(G_1)=\ord(G')$, since $G'$ contains $t$ more vertices and
$t$ more edges than $G_1$.

Since $v_t$ is a leaf in $G'$, and it cannot be a leaf in $G$,
it follows that $G$ has an edge that is not in $G'$; hence
$\ord(G)>\ord(G')$.  But $\ord(G')=\ord(G_1)$.
\end{proof}

\begin{proof}[Proof of Lemma~\ref{le_finiteness_and_positivity}]
Let $\psi$ be an element of
$\Psi^+_{<r}$.  Then there is a sequence of strictly increasing graphs
$$
\psi_0\subset\psi_1\subset\cdots\subset \psi_t=\psi
$$
where $\psi_0$ is isomorphic to
some $\psi_i$, and for each $j\in[t]$, $\psi_j$ is the union of
$\psi_{j-1}$ and a graph $\widetilde\psi_j$ that is isomorphic to
one of the $\psi_i$.
It follows that each $\psi_j$ is pruned; the lemma about increasing
order of pruned graphs shows that $\ord(\psi_t)>t+\ord(\psi_0)\ge t+1$.
Hence $t\le r-2$, and hence $\psi$ is the union of $r-1$ graphs,
each of which lies in a class in $\Psi$.  Hence
$$
\#E_\psi \le (r-1) \max_{i\in[t]} (\#E_{\psi_i})
$$
which is bounded.  Since $\psi$ is positive,
$$
\#V_\psi \le \#E_\psi - \ord(\psi) \le \#E_\psi
$$
so $\#V_\psi$ is bounded.  Hence there are only finitely many possible
isomorphism classes of graphs, $\psi$, as graphs, and hence only
finitely many possible $B$-graph classes in $\Psi^+_{<r}$.
\end{proof}

\subsection{Proof of Lemma~\ref{le_morphisms_factor}}

\begin{proof}[Proof of Lemma~\ref{le_morphisms_factor}]
The first claim is easy: if $[\psi_\Bg]\in\Psi^+$, then $\psi_\Bg$ is
the union of injective morphisms $\psi_\Bg^i\to\psi_\Bg$; hence
the image 
of any injective
map $u\from \psi_\Bg\to G_\Bg$ is the union of the images of the compositions
$\psi_\Bg^i\to\psi_\Bg$ with $u$, which are injective morphisms.
Since $\Psi^+\cap G$ contains all of these images, it contains
$u(\Psi_\Bg)$.  Hence $u$ factors through $\Psi^+\cap G$.

For the second claim, let $G^0_\Bg$ be the image of the injective
morphism $\psi_\Bg\to G_\Bg$.
Since $\Psi^+\cap G_\Bg$ is the union of images
of injective maps from the $\psi_\Bg^i$ to $G_\Bg$, there must be
a strict inclusion of graphs
$$
G^0_\Bg\subset G^1_\Bg \subset \cdots\subset G^t_\Bg=\Psi^+\cap G_\Bg
$$
where each
$G^i_\Bg$ is the union of $G^{i-1}_\Bg$ and a $B$-subgraph of 
$G_\Bg$ that lies in some element of $\Psi$.
Let $i$ be the smallest value such that $\ord(G^i)\ge r$; then $i\ge 1$
(since $\ord(\psi)<r$ by assumption and
$G^0_\Bg$ and $\psi_\Bg$ are isomorphic, so $\ord(G^0)<r$). 
Since $G^i$ has at most $s$ more edges than $G^{i-1}$, we have
$$
\ord(G^i) \le \ord(G^{i-1}) + s\le r-1+s.
$$
Hence the injection $\psi_\Bg\to G$ factors through $\psi'_\Bg\eqdef G^i_\Bg$,
and
$$
\psi'_\Bg\in \Psi^+_{<r+s}\setminus \Psi^+_{<r}
$$
as desired;
\eqref{eq_factor_inequality} follows because the number of 
injective morphisms $\psi_\Bg\to G_\Bg$ that factor through $\psi'_\Bg$
is at most the number of injective morphisms $\psi_\Bg\to\psi'_\Bg$ times
those $\psi'_\Bg\to G_\Bg$.
\end{proof}

\subsection{The Injection Count and Resulting Partial Order}

In this section we describe the partial order which we will later
use to define the rational numbers
$\mu[\psi_\Bg]$ that are fundamental to our
construction of ``approximate indicator functions'' (see 
\eqref{eq_approximate_indicator}).

\begin{definition}\label{de_poset_B_graphs}
Let $B$ be a graph.
Recall that $[S_\Bg]$ denotes the class of $B$-graphs
that are isomorphic to $S_\Bg$; write
$[S_\Bg]\le_B [T_\Bg]$ if $N(S_\Bg,T_\Bg)>0$ 
(Definition~\ref{de_injective_morphisms}), i.e., if there exists
an injection $S_\Bg\to T_\Bg$;
we sometimes refer to $\le_B$ as 
$\le$ when confusion is unlikely to occur.
\end{definition}

\begin{lemma}
Let $B$ be a graph.
Then the relation $\le_B$ in Definition~\ref{de_poset_B_graphs}
is a partial order (of isomorphism classes of $B$-graphs).
\end{lemma}
\begin{proof}
The relation $\le=\le_B$ is clearly reflexive and transitive, so we need 
only show that it is anti-symmetric: so assume that
$[S_\Bg]\le [T_\Bg]$ and $[T_\Bg]\le [S_\Bg]$; let us
prove that
$[S_\Bg]= [T_\Bg]$.  Since $[S_\Bg]\le [T_\Bg]$,
there is a morphism $\nu\from S_\Bg\to T_\Bg$ that is injective;
hence $\#V_S\le \#V_T$ and $\#\Edir_S\le\#\Edir_T$;
then $[T_\Bg]\le [S_\Bg]$ provides the reverse inequalities, and hence
$\#V_S=\#V_T$ and $\#\Edir_S=\#\Edir_T$.
It follows that
$\nu$ is bijective on the vertex sets and directed edge sets, and
we easily check that the inverse map on these sets yields an isomorphism
of $B$-graphs.  Hence $S_\Bg$ and $T_\Bg$ are isomorphic,
and hence $[S_\Bg]= [T_\Bg]$.
\end{proof}

\subsection{The M\"obius Function}

Now let $B$ be a graph, and
let $\cO$ be any subset of the set of isomorphism classes of
$B$-graphs;
then the partial order $\le=\le_B$ above restricts to give
a partial order on $\cO$.
We now define a {\em M\"obius function} for the partially ordered
set, $\cO$, in
the usual way: we define a function $\mu=\mu_\cO\from\cO\to\reals$
with the property that
\begin{equation}\label{eq_mobius}
\sum_{\substack{[S_\Bg]\in\cO \\ [S_\Bg]\le [T_\Bg]}} 
N([S_\Bg],[T_\Bg])
\mu_\cO([S_\Bg]) = 1
\end{equation}
for each $[T_\Bg]$, by defining 
$$
\mu([T_\Bg])=1/N([T_\Bg],[T_\Bg]) = 1 / \bigl(\#{\rm Aut}(T_\Bg) \bigr)
$$
when $[T_\Bg]$ is a minimal element of $\cO$, and then---by
structural induction on $T$, or regular
induction on $(\#V_T)+(\#\Edir_T)$---we set
\begin{equation}\label{eq_mobius_inductive}
\mu_\cO([T_\Bg]) = 
\frac{1}{N([T_\Bg],[T_\Bg])}
\Biggl( 1 - 
\sum_{\substack{[S_\Bg]\in\cO \\ [S_\Bg]< [T_\Bg]}} 
N([S_\Bg],[T_\Bg])
\mu_\cO([S_\Bg])
\Biggr).
\end{equation} 

Notice that $N(S_\Bg,S_\Bg)$ can be strictly greater than one
(e.g., which can happen if $S$ is a cycle, or a ``barbell graph'' or
``theta graph''; see Section~6 of Article~I).
For this reason, the $\mu_\cO([S_\Bg])$ 
are not necessarily integers.

\begin{definition}
Let $B$ be a graph, and $\cO$ a subset of the set of isomorphism
classes of $B$-graphs.
By the {\em M\"obius function on $\cO$} we 
mean the unique function $\mu([S_\Bg])=\mu_\cO([S_\Bg])$ satisfying
\eqref{eq_mobius},
defined inductively by 
\eqref{eq_mobius_inductive}.
\end{definition}

We remark that \cite{friedman_alon} works with $B$-graphs (and graphs); 
here we work
with $B$-graphs simply because our definition of algebraic model
makes this convenient.

\subsection{The Truncated Indicator Function and the Proof
of Theorem~\ref{th_mobius_inversion_truncation}}

\begin{definition}
Let $B$ be a graph, and $\Psi$ a set of isomorphism classes of $B$-graphs.
Then $\Psi^+$ is a set of isomorphism classes of $B$-graphs, and
gives rise to a M\"obius function $\mu_{\Psi^+}$.
For each $r$ we define the {\em order $r$ truncated $\Psi$-indicator function}
to be the function defined on $B$-graphs, $G_\Bg$, given by
\begin{equation}\label{eq_approximate_indicator}
I_r(\Psi,G_\Bg) \eqdef 
\sum_{[S_\Bg]\in \Psi^+_{<r}} N(S_\Bg,G_\Bg) \mu_{\Psi^+}([S_\Bg]) .
\end{equation}
\end{definition}
This sum is finite for any given $G_\Bg$ since $N(S_\Bg,G_\Bg)=0$ if
$S$ has more vertices than $G$; in cases of interest to us we will
require $\Psi^+_{<r}$ to be finite for all $r$, so that the above sum
involves finitely many $[S_\Bg]$.

\begin{proof}[Proof of Theorem~\ref{th_mobius_inversion_truncation}]
Let $\mu[S_\Bg]=\mu_{\Psi^+}[S_\Bg]$ be the M\"obius function
for the partially ordered set $\Psi^+$.
Let us prove the various claims in 
Theorem~\ref{th_mobius_inversion_truncation} regarding
$I_r(\Psi,G_\Bg)$ for any $B$-graph $G_\Bg$.

Set $\psi_\Bg=\Psi^+\cap G_\Bg$, and consider the three cases
where $\ord_\Psi(G)=\ord(\psi)$ is $0$, between $1$ and $r-1$, and
at least $r$.
In all cases it will be useful to note that by 
Lemma~\ref{le_morphisms_factor}, for all $[S_\Bg]\in\Psi^+$
we have that
$$
N(S_\Bg,G_\Bg) = N(S_\Bg,\Psi^+\cap G_\Bg)=N(S_\Bg,\psi_\Bg) ,
$$
and hence $I_r(\Psi,G_\Bg)=I_r(\Psi,\psi_\Bg)$.

First, consider the case where
$\ord(\psi)=0$.  Then
$\Psi^+\cap G_\Bg=\emptyset_\Bg$, the empty graph, and
$N(S_\Bg,G_\Bg)=N(S_\Bg,\psi_\Bg)=0$ for all $S_\Bg\in\Psi^+$, and hence
$$
I_r(\Psi,G_\Bg)  = 0 = \II_{{\rm Meets}(\Psi)}(G).
$$
This proves \eqref{eq_trunc_exact} in this case.

Second, consider the case where $1\le \ord(\psi)<r$, and
hence $\II_{{\rm Meets}(\Psi)}(G)=1$.
By Lemma~\ref{le_pruned_inclusions}, if $[\psi_\Bg]\in \Psi^+$
and $\ord(\psi)<r$, then for $S_\Bg\in\Psi^+$ with $\ord(S)\ge r$
we have
$$
N(S_\Bg,\psi_\Bg) = 0,
$$
and hence
\begin{align*}
I_r(\Psi,G_\Bg)=I_r(\Psi,\psi_\Bg) & =
\sum_{[S_\Bg]\in \Psi^+_{<r}} N(S_\Bg,\psi_\Bg) \mu[S_\Bg] 
\\
& = 
\sum_{[S_\Bg]\in \Psi^+} N(S_\Bg,\psi_\Bg) \mu[S_\Bg]
= 1 = \II_{{\rm Meets}(\Psi)}(G_\Bg)
\end{align*}
by \eqref{eq_mobius}.
This proves \eqref{eq_trunc_exact} in this case.

Third and lastly, consider the case that $\ord(\psi)\ge r$.
Then \eqref{eq_factor_inequality} of 
Lemma~\ref{le_morphisms_factor} implies that
for any $[S_\Bg]\in\Psi^+$ with $\ord(S)<r$ we have
$$
N(S_\Bg,G_\Bg) \le \sum_{\psi'_\Bg\in \Psi^+_{<r+s}\setminus \Psi^+_{<r}}
N(S_\Bg,\psi'_\Bg)
N(\psi'_\Bg,B_\Bg).
$$
This implies \eqref{eq_reasonable_trunc_bound} for
$$
C = \max_{\psi'_\Bg \in \Psi^+_{<r+s}\setminus \Psi^+_{<r}}\ \ 
\sum_{[S_\Bg]\in\Psi^+_{<r}}  \bigl| \mu[S_\Bg] \bigr| \, N(S_\Bg,\psi'_\Bg) ,
$$
which is a finite sum since $\Psi^+_{<r}$ and $\Psi^+_{<r+s}$
are finite sets by assumption.
\end{proof}

\section{Proof of Theorem~\ref{th_exp_ind_cert}}
\label{se_ind_cert_proof}

In this section we prove Theorem~\ref{th_exp_ind_cert}.  We build up
the proof with a sequence of lemmas; we start
the following lemma based on Theorem~\ref{th_main_certified_pairs}.

\begin{lemma}\label{le_main_pair_exp_conclusion}
Let $B$ be a graph, and $\Psi$ a finite family of $B$-graphs such that
$\Psi^+_{<s}$ is finite for all $s\in\naturals$.
Then for any fixed $\nu>0$ and $r',r''\in\naturals$ we have that
$$
f(k,n)=\EE_{G\in\cC}[I_{r'}(\Psi,G_\Bg){\rm cert}_{<\nu,<r''}(G,k)]
$$
has a $(B,\nu)$-Ramanujan expansion to any order $r\in\naturals$
$$
c_0(k)+\cdots+c_{r-1}(k)/n^{r-1}+ O(1) c_r(k)/n^r,
$$
where the bases of the $c_i=c_i(k)$ are a subset of any set
of eigenvalues of the model,
and where $c_i(k)=0$ provided that $i<r$ and $i$ is less than the order
of any $B$-graph occurring in $\cC_n(B)$ that meets $\Psi$, i.e.,
that lies in an element of $\Psi$.
\end{lemma}
As usual, the same lemma holds for the weakly-certified trace
${\rm cert}_{\le\nu,<r''}(G,k)$, but we shall not need this result.
\begin{proof}
Since $\Psi^+_{<r'}$ is finite,
\eqref{eq_approximate_indicator} 
implies that $I_{r'}(\Psi,G_\Bg)$ is a finite linear combination of
functions $G_\Bg\mapsto N(\psi_\Bg,G_\Bg)$;
hence to prove this lemma, it suffices to prove such expansions
exist for functions of the form
$$
f(k,n)=\EE_{G\in\cC}[ N(\psi_\Bg,G_\Bg) {\rm cert}_{<\nu,<r''}(G,k)] 
$$
with $[\psi_\Bg]\in\Psi^+_{<r'}$.
By subdividing the certified walks
by their homotopy type and applying
\eqref{eq_inclusion_exclusion_M}, 
it suffices to prove such expansions
exist for functions of the form
\begin{align*}
f(k,n) & =
\EE_{G\in\cC}[ N(\psi_\Bg,G_\Bg) \snbc(T^{\og},\ge \bec\xi;G,k)] \\
& = 
\EE_{G\in\cC}[ \#([\psi_\Bg]\cap G_\Bg) \snbc(T^{\og},\ge \bec\xi;G,k)] 
\end{align*}
for a fixed $[\psi_\Bg]\in\Psi^+_{<r}$,
a fixed ordered graph, $T^{\og}$, and a fixed $\bec\xi$ with
\begin{equation}\label{eq_vlg_xi_nu_bound}
\mu_1 \bigl( \VLG(T,\bec\xi) \bigr) <\nu.
\end{equation} 
But this follows from Theorem~\ref{th_main_certified_pairs}.
\end{proof}


\begin{lemma}
\label{le_large_ind_order}
Let $B$ be a graph and $\Psi$ a collection of finite isomorphism classes
of $B$-graphs such that
$\Psi^+_{<r}$ is finite for all $r$.  
Let $\II_{\ord_\Psi^{-1}(\ge r)}(G_\Bg)$ denote
the indicator function
of those $B$-graphs $G_\Bg$ with $\ord_\Psi(G)\ge r$.
Then for any algebraic model $\cC_n(B)$ and $r\in\naturals$
there is a $C$ such that 
\begin{align*}
\EE_{G\in\cC_n(B)}[\II_{\ord_\Psi^{-1}(\ge r)}(G_\Bg)
{\rm cert}_{<\nu,<r}(G,k)]  
& \le C n^{-r+1}\mu_1(B)^k
\\
\EE_{G\in\cC_n(B)}[\II_{\ord_\Psi^{-1}(\ge r)}(G_\Bg)
I(G,\Psi,r) 
{\rm cert}_{<\nu,<r}(G,k)] 
& \le C n^{-r+1}\mu_1(B)^k 
\\
\EE_{G\in\cC_n(B)}[\II_{\ord_\Psi^{-1}(\ge r)}(G_\Bg)
I(G,\Psi,r) ]
& \le C n^{-r}
\end{align*}
for all $k\in\naturals$.
\end{lemma}
\begin{proof}
For any graph $G$ we have
$$
0\le {\rm cert}_{<\nu,<r}(G,k) \le {\rm snbc}(G,k)\le \Trace(H_G^k)
\le n \Trace(H_B^k) \le  n \mu_1^k(B) (\#\Edir_B) .
$$
Hence it suffices to show that both
\begin{equation}\label{eq_prob_bound_large_order}
\EE_{G\in\cC_n(B)}[\II_{\ord_\Psi(\ge r)}(G_\Bg)]
\quad \mbox{i.e., 
$\Prob_{G\in\cC_n(B)}[\ord_\Psi(G_\Bg)\ge r]$},
\end{equation} 
and
\begin{equation}\label{eq_indicator_approx_with_large_order}
\EE_{G\in\cC_n(B)}[\II_{\ord_\Psi(\,\cdot\,)\ge r}(G_\Bg)
I_{r}(\Psi,G_\Bg)] 
\end{equation} 
are bounded by $O(n^{-r})$.  

By Lemma~\ref{le_morphisms_factor},
if $\ord_\Psi(G_\Bg)\ge r$ then $N(\psi_\Bg,G_\Bg)\ge 1$ for some
$\psi_\Bg\in\Psi^+_{<r+s}\setminus\Psi^+_{<r}$.  Hence
\begin{align*}
\EE_{G\in\cC_n(B)}[\II_{\ord_\Psi(\,\cdot\,)\ge r}(G_\Bg)]
& \le
\sum_{[\psi_\Bg]\in\Psi^+_{<r+s}\setminus\Psi^+_{<r}}
\Prob_{G\in\cC_n(B)}[ N(\psi_\Bg,G_\Bg) > 0 ]
\\
& \le
\sum_{[\psi_\Bg]\in\Psi^+_{<r+s}\setminus\Psi^+_{<r}}
\EE_{G\in\cC_n(B)} [N(\psi_\Bg,G_\Bg)] .
\end{align*}
Since the above sum is over finitely many classes $\psi_\Bg$, and
since for each $\psi_\Bg$ we have that
\begin{equation}\label{eq_reusable_inequality}
\EE_{G\in\cC_n(B)}[ N(\psi_\Bg,G_\Bg) ] =
\EE_{G\in\cC_n(B)}[ \#([\psi_\Bg^\og]\cap G_\Bg)] = O(n^{-\ord(\psi)})
= O(n^{-r})
\end{equation} 
since $\cC_n(B)$ is algebraic, we get the desired bound on 
\eqref{eq_prob_bound_large_order}.

To get the desired bound on \eqref{eq_indicator_approx_with_large_order},
since $I_r(\Psi,G_\Bg)$ is a linear combination of the functions
$N(\psi'_\Bg,G_\Bg)$ with $\psi'_\Bg\in\Psi^+_{<r}$, it suffices to prove such
a bound for each function
$$
\EE_{G\in\cC_n(B)}[\II_{\ord_\Psi^{-1}(\ge r)}(G_\Bg)N(\psi'_\Bg,G_\Bg)] \ .
$$
But if $\ord_\Psi(G)\ge r$, then we have
$$
N(\psi'_\Bg,G_\Bg) \le C 
\sum_{[\psi_\Bg]\in\Psi^+_{<r+s}\setminus\Psi^+_{<r}}
\EE_{G\in\cC_n(B)}[ N(\psi_\Bg,G_\Bg) ]
$$
according to \eqref{eq_factor_inequality}, where $C$ is the maximum
value of $N(\psi'_\Bg,\psi_\Bg)$ over 
$[\psi_\Bg]\in\Psi^+_{<r+s}\setminus\Psi^+_{<r}$;
hence it suffices to bound
$$
\EE_{G\in\cC_n(B)}[ N(\psi_\Bg,G_\Bg) ]
$$
for each $\psi_\Bg$ with $[\psi_\Bg]\in \Psi^+_{<r+s}\setminus\Psi^+_{<r}$, 
which again is implied by
\eqref{eq_reusable_inequality}.
\end{proof}

We will see that the following lemma almost immediately implies
Theorem~\ref{th_exp_ind_cert}.

\begin{lemma}\label{le_stronger}
Let $B$ be a graph, and $\Psi$ a finite family of $B$-graphs such that
$\Psi^+_{<s}$ is finite for all $s\in\naturals$.
Then for any fixed $\nu>0$ and $r''\in\naturals$ we have that
\begin{equation}\label{eq_meets_times_cert}
f(k,n)=\EE_{G\in\cC_n(B)}
[\II_{{\rm Meets}(\Psi)}(G_\Bg){\rm cert}_{<\nu,<r''}(G,k)]
\end{equation} 
has a $(B,\nu)$-Ramanujan expansion to any order $r\in\naturals$
$$
c_0(k)+\cdots+c_{r-1}(k)/n^{r-1}+ O(1) c_r(k)/n^r,
$$
where the bases of the $c_i=c_i(k)$ are the exponents of the model,
and where $c_i(k)=0$ provided that $i$ is less than the order
of any $B$-graph occurring in $\cC_n(B)$ that meets $\Psi$, i.e.,
that lies in an element of $\Psi$.
\end{lemma}
\begin{proof}
Let $r'=r+1$.  We have
\begin{equation}\label{eq_indicator_Psi_inv_sum}
1 = 
\II_{\ord_\Psi^{-1}(\ge r')}(G_\Bg) + 
\II_{\ord_\Psi^{-1}(< r')}(G_\Bg) 
\end{equation} 
where $\ord_\Psi^{-1}(\ge r')$ and $\ord_\Psi^{-1}(<r')$ denote the set
of $B$-graphs, $G_\Bg$, for which $\ord_\Psi(G_\Bg)$ is, respectively,
$\ge r'$ and $<r'$.
Using the facts that
\begin{enumerate}
\item $\ord_\Psi(G_\Bg)\ge r'$ implies that $G\in{\rm Meets}(\Psi)$, and
\item $\ord_\Psi(G_\Bg)<r'$ implies that 
$I_{r'}(\Psi,G_\Bg)=\II_{{\rm Meets}(\Psi)}(G_\Bg)$,
\end{enumerate}
we have
\begin{align}
\label{eq_meets_on_lhs}
\II_{{\rm Meets}(\Psi)}(G_\Bg) 
& =
\II_{{\rm Meets}(\Psi)}(G) \II_{\ord_\Psi^{-1}(< r')}(G_\Bg)
+ \II_{{\rm Meets}(\Psi)}(G)  \II_{\ord_\Psi^{-1}(\ge r')}(G_\Bg) 
\\
\nonumber
& =
I_{r'}(\Psi,G_\Bg) \II_{\ord_\Psi^{-1}(< r')}(G_\Bg)
+ \II_{\ord_\Psi^{-1}(\ge r')}(G_\Bg)  
\\
\label{eq_terms_we_like}
& =
I_{r'}(\Psi,G_\Bg) 
-
I_{r'}(\Psi,G_\Bg) \II_{\ord_\Psi^{-1}(\ge r')}(G_\Bg) 
+ \II_{\ord_\Psi^{-1}(\ge r')}(G_\Bg)  
\end{align}
Now we multiply the left-hand-side of \eqref{eq_meets_on_lhs} and
\eqref{eq_terms_we_like} by ${\rm cert}_{<\nu,<r''}(G,k)$ and taking
expected values: the left-hand-side becomes \eqref{eq_meets_times_cert},
and the individual summands of \eqref{eq_terms_we_like} become a sum of
\begin{equation}\label{eq_first_summand_in_mess}
\EE_{G\in\cC_n(B)}[I_{r'}(\Psi,G_\Bg){\rm cert}_{<\nu,<r''}(G,k)]
\end{equation}
plus terms bounded by $C n^{-r'+1}\mu_1(B)^k=Cn^{-r}\mu_1(B)$
by Lemma~\ref{le_large_ind_order}.
By Lemma~\ref{le_main_pair_exp_conclusion}, 
\eqref{eq_first_summand_in_mess} has an $(B,\nu)$-bounded order $r$
expansion, and hence so does
\eqref{eq_meets_times_cert}.
\end{proof}

\begin{proof}[Proof of Theorem~\ref{th_exp_ind_cert}]
If $\cT$ is a set of $B$-graphs, then let $\Psi$ be the isomorphism
classes of a finite set of positive generators of $\cT$; if
$\cT$ is a set of graphs, let $\Psi$
consist of all possible $B$-graph structures on
elements of a finite set of positive generators of $\cT$.

Then for all $G=G_\Bg\in\cC_n(B)$, 
$G$ meets $\cT$ iff $G_\Bg$ meets $\Psi$.
According to Lemma~\ref{le_finiteness_and_positivity}, 
$\Psi^+_{<r}$ is finite for all $r$.
Now we apply Lemma~\ref{le_main_pair_exp_conclusion}.
\end{proof}

\section{Proof of Theorem~\ref{th_exp_ind}}
\label{se_ind_proof}

In this section we easily prove Theorem~\ref{th_exp_ind} based
on the methods we have already developed.

\begin{proof}[Proof of Theorem~\ref{th_exp_ind}]
If $\cT$ is a set of $B$-graphs, then let $\Psi$ be the isomorphism
classes of a finite set of positive generators of $\cT$; if
$\cT$ is a set of graphs, let $\Psi$
consist of all possible $B$-graph structures on
elements of a finite set of positive generators of $\cT$.
Hence ${\rm Meets}(\cT)={\rm Meets}(\Psi)$.

Fix an $r\in\naturals$.
Let $s$ be the largest number of edges in a graph in a class of
$\Psi$.  According to Lemma~\ref{le_finiteness_and_positivity},
$\Psi^+_{r+s}$ is finite.

With notation as in \eqref{eq_indicator_Psi_inv_sum}, we have
\begin{align*}
\II_{{\rm Meets}(\Psi)}(G) 
& = 
\II_{{\rm Meets}(\Psi)}(G) 
\II_{\ord_\Psi^{-1}(< r)}(G_\Bg)  +
\II_{{\rm Meets}(\Psi)}(G) 
\II_{\ord_\Psi^{-1}(\ge r)}(G_\Bg) 
\\
& =
\II_{{\rm Meets}(\Psi)}(G)
\II_{\ord_\Psi^{-1}(< r)}(G_\Bg)  +
\II_{\ord_\Psi^{-1}(\ge r)}(G_\Bg)
\end{align*}
since $\ord_\Psi(G_\Bg)\ge r$ implies that $G_\Bg$ meets $\Psi$.
With $I_r(\Psi,G_\Bg)$ as in \eqref{eq_approximate_indicator}, we have
$$
\II_{{\rm Meets}(\Psi)}(G)
\II_{\ord_\Psi^{-1}(< r)}(G_\Bg)
=
I_r(\Psi,G_\Bg)
\II_{\ord_\Psi^{-1}(< r)}(G_\Bg),
$$
since $I_r(\Psi,G_\Bg)=\II_{{\rm Meets}(\Psi)}(G)$ provided that
$\ord_\Psi(G_\Bg)<r$.
Combining the above two displayed equations we have
$$
\II_{{\rm Meets}(\Psi)}(G) =
I_r(\Psi,G_\Bg)
\II_{\ord_\Psi^{-1}(< r)}(G_\Bg)  +
\II_{\ord_\Psi^{-1}(\ge r)}(G_\Bg)
$$
Taking expected values yields
\begin{equation}\label{eq_meets_and_I_r}
\EE_{G\in\cC_n(B)}[\II_{{\rm Meets}(\Psi)}(G)] =
\EE_{G\in\cC_n(B)}\bigl[
I_r(\Psi,G_\Bg)
\II_{\ord_\Psi^{-1}(< r)}(G_\Bg) 
\bigr] 
+
O(n^{-r})
\end{equation} 
in view of Lemma~\ref{le_large_ind_order}.
In view of the fact that $\Psi^+_{<r+s}\setminus\Psi^+_{<r}$ is finite,
taking expected values in \eqref{eq_reasonable_trunc_bound} yields
$$
\EE_{G\in\cC_n(B)}\bigl[
I_r(\Psi,G_\Bg)
\II_{\ord_\Psi^{-1}(\ge r)}(G_\Bg) 
\bigr] 
\le
\sum_{[\psi_\Bg]\in \Psi^+_{<r+s}\setminus\Psi^+_{<r}} O(n^{-\ord(\psi)})
=O(n^{-r}).
$$
Adding this to \eqref{eq_meets_and_I_r} yields
\begin{align*}
& \EE_{G\in\cC_n(B)}[\II_{{\rm Meets}(\Psi)}(G)]  + O(n^{-r})
\\
& = 
\EE_{G\in\cC_n(B)}\bigl[
I_r(\Psi,G_\Bg)
\II_{\ord_\Psi^{-1}(< r)}(G_\Bg) 
\bigr] 
+
\EE_{G\in\cC_n(B)}\bigl[
I_r(\Psi,G_\Bg)
\II_{\ord_\Psi^{-1}(\ge r)}(G_\Bg) 
\bigr] 
+ O(n^{-r}),
\end{align*}
and, since $\II_{\ord_\Psi^{-1}(< r)}+\II_{\ord_\Psi^{-1}(\ge r)}=1$,
\begin{equation}
\label{eq_meets_I_r_within_ord_r}
\EE_{G\in\cC_n(B)}[\II_{{\rm Meets}(\Psi)}(G)]  =
\EE_{G\in\cC_n(B)}\bigl[
I_r(\Psi,G_\Bg)
\bigr]  + O(n^{-r}).
\end{equation} 

Next consider
$$
\EE_{G\in\cC_n(B)}[I_r(\Psi,G_\Bg)] .
$$
For
each $S_\Bg\in\Psi^+_{<r}$ we have
\begin{equation}\label{eq_expected_N}
\EE_{G\in\cC_n(B)}[N(S_\Bg,G_\Bg)] = 
\EE_{G\in\cC_n(B)}[ \#[S_\Bg]\cap G_\Bg ] =
c_0+c_1/n+\ldots+c_{r-1}/n^{r-1}+O(1/n^r),
\end{equation} 
with $c_i=0$ for $i<\ord(S)$ since $\cC_n(B)$ is
algebraic; furthermore $c_i>0$ if $i=\ord(S)$ and
$S_\Bg$ occurs in $\cC_n(B)$.  Since $\Psi^+_{<r}$ is finite and 
$I_r(\Psi,G_\Bg)$ is, by definition, a finite linear combination of
functions of the form \eqref{eq_expected_N}, we have
\begin{equation}\label{eq_trunc_indic_has_series}
\EE_{G\in\cC_n(B)}[I_r(\Psi,G_\Bg)]
=
c_0+c_1/n+\ldots+c_{r-1}/n^{r-1}+O(1/n^r),
\end{equation}
where if $j$ is the minimum order of an element of $\Psi^+_{<r}$
occurring in $\cC_n(B)$, then
with $c_i=0$ for $i<j$ (assuming that $r>i$ so that
$c_i$ is uniquely defined).

Now let $j$ be the minimum order of an element of $\Psi^+$ occurring
in $\cC_n(B)$.
By Lemma~\ref{le_pruned_inclusions}, if
$[S_\Bg]\in\Psi^+$ occurs in $\cC_n(B)$ 
and $\ord(S)=j$, then
$[S_\Bg]$ is a minimal element of the partially
ordered set $\Psi^+$.
[This implies that $[S_\Bg]\in\Psi$, but this is inconsequential here.]
It follows that 
$$
\mu[S_\Bg] = 1 / \bigl( \# {\rm Aut}(S_\Bg)\bigr) > 0.
$$
It follows that $c_j$ in \eqref{eq_trunc_indic_has_series} is
the sum of $\mu[S_\Bg]$ over all such $S_\Bg$, and hence
$c_j>0$ (again, assuming $r>j$ so that $c_j$ is uniquely determined).

Combining \eqref{eq_trunc_indic_has_series} with
\eqref{eq_meets_I_r_within_ord_r}
yields
$$
\EE_{G\in\cC_n(B)}[\II_{{\rm Meets}(\Psi)}(G)]  =
c_0+c_1/n+\ldots+c_{r-1}/n^{r-1}+O(1/n^r),
$$
for the same $c_i$ as in \eqref{eq_trunc_indic_has_series},
as desired.
%
%
\end{proof}

\section{Conclusion of the Proofs of Theorems~\ref{th_main_two_results}
and \ref{th_extra_needed}}
\label{se_finish_proof}

In this section we complete the proof
of Theorem~\ref{th_main_two_results}.
We will need the following straightforward fact.

\begin{lemma}\label{le_min_tang_ord_pos}
Let $\nu>1$ and $r\in\naturals$,
and let $\psi$ be a minimal $(\ge\nu,<r)$-tangle
(Definition~\ref{de_minimal_tangle},
i.e., $\psi$ is a $(\ge\nu,<r)$-tangle but no proper subgraph of
$\psi$ is).
Then $\ord(\psi)\ge 1$.
\end{lemma}
\begin{proof}
By definition a tangle is necessarily 
connected, and hence $\psi$ is connected. 
Let us show that $\psi$ is pruned:
if $v\in V_\psi$, $v$ cannot be of degree one, incident upon
an edge, $e$, joining $v$ to some other vertex, $w$, for then we could
``prune'' $\psi$, discarding $v$ and $e$; this pruning preserves the
$\mu_1$ and the order, and does not affect the connectedness,
which would contradict the minimality of
$\psi$.
Furthermore, $v$ cannot be of degree one incident upon a half-loop,
or of degree zero, since in either case $V_\psi=\{v\}$ (since $\psi$
is connected), and then $\mu_1(\psi)$, in either case, is $0$.
Hence each vertex of $\psi$ is of degree at least two, and so $\psi$
is pruned.

Since $\mu_1$ of a cycle of any length equals $1$, and $\nu>1$,
$\psi$ cannot be a cycle.  Hence
Lemma~\ref{le_pruned_cycle_or_pos_ord} shows that $\ord(\psi)\ge 1$.
\end{proof}

%

\begin{proof}[Proof of Theorem~\ref{th_main_two_results}]
Let $\cT={\rm HasTangles}(\ge\nu,<r')$.
Each element of $\cT$ contains an $(\ge\nu,<r')$-tangle, and 
by Theorem~\ref{th_tangle_finite} there are finitely many such
tangles, up to isomorphism, that are minimal with respect to
inclusion.  
According to Lemma~\ref{le_min_tang_ord_pos}, each minimal tangle
has order at least one, and by definition any tangle is connected.
Hence
$\cT$ is finitely positively generated.

Hence we can apply 
Theorems~\ref{th_exp_ind} and~\ref{th_exp_ind_cert}
to $\cT={\rm HasTangles}(\ge\nu,<r)$.  
Theorem~\ref{th_exp_ind} implies that
$$
\Prob_{\cC_n(B)}[{\rm HasTangles}(\ge\nu,<r')] 
$$
has an expansion to any order $r$
$$
c_0 + c_1/n + \cdots + c_{r-1}/n^{r-1} + O(1/n^r)
$$
where $c_i=0$ for $i<r$ and 
$i<i_0$ where $i_0$ is the smallest order of a
$(\ge\nu)$-tangle.  Since $\nu>1$, 
Lemma~\ref{le_pruned_cycle_or_pos_ord} implies that $i_0\ge 1$.
Taking $r\ge i_0$ (which we are free to do),
since $c_0=0$ and
\begin{equation}\label{eq_tangle_free_has_tangles}
\II_{\rm TangleFree(\ge\nu,<r')} = 
1 - 
\II_{\rm HasTangles(\ge\nu,<r')},
\end{equation} 
it follows that
$$
\Prob_{\cC_n(B)}[{\rm TangleFree(\ge\nu,<r')}]
= 1 - c_1/n - \cdots - c_{r-1}/n^{r-1} + O(1/n^r)
$$
which, along with $c_i=0$ for $i<i_0$, establishes
the part of Theorem~\ref{th_main_two_results} regarding
the asymptotic expansions (and their coefficients) for
\eqref{eq_main_tech_result2}.

Similarly we use \eqref{eq_tangle_free_has_tangles},
and subtract the result in
Theorem~\ref{th_exp_ind_cert}
from than of
Theorem~\ref{th_main_tech_result_for_cert}
to obtain that for any $r',r''$
$$
\EE_{G\in\cC_n(B)}[ \II_{{\rm TangleFree}(\ge\nu,<r')}(G) 
{\rm cert}_{<\nu,r''}(G) ]
$$
has an expansion to any order $r$, whose coefficients $c_i(k)$
have $c_0(k)$ as in \eqref{eq_give_c_zero_k},
since the $1/n^i$-coefficients in
Theorem~\ref{th_exp_ind_cert} vanish for $i<r$ and
$i$ less than the smallest order of a $(\ge\nu)$-tangle,
which is at least $1$.
Now we take
$r=r'=r''\ge 1$; according to 
\eqref{eq_certified_versus_snbc}
we have
$$
\II_{{\rm TangleFree}(\ge\nu,<r)}(G)
{\rm cert}_{<\nu,<r}(G,k)
=
\II_{{\rm TangleFree}(\ge\nu,<r)}(G)
{\rm snbc}_{<r}(G,k) ,
$$
whereupon we have that
\begin{equation}\label{eq_less_than_r_done}
\EE_{G\in\cC_n(B)}[ \II_{{\rm TangleFree}(\ge\nu,<r)}(G) 
{\rm snbc}_{<r}(G,k) ]
\end{equation} 
has an expansion to order $r$ with $c_0(k)$ as in \eqref{eq_give_c_zero_k}.
Finally we note that 
\begin{equation}\label{eq_r_or_greater_done} 
\EE_{G\in\cC_n(B)}[ \II_{{\rm TangleFree}(\ge\nu,<r)}(G) 
{\rm snbc}_{\ge r}(G,k) ]
\le
\EE_{G\in\cC_n(B)}[ {\rm snbc}_{\ge r}(G,k) ] \le g(k) O(1)/n^r
\end{equation} 
where $g$ is a function of growth $\mu_1(B)$.
Adding~\eqref{eq_less_than_r_done}
and \eqref{eq_r_or_greater_done}, and using
$\snbc(G,k)=\Trace(H_G^k)$
establishes the claim in 
Theorem~\ref{th_main_two_results} for
asymptotic expansions of
\eqref{eq_main_tech_result1}.
\end{proof}

\begin{proof}[Proof of Theorem~\ref{th_extra_needed}]
We apply Theorem~\ref{th_exp_ind} with $\cT=[S_\Bg]$, which
is positively generated since $S$ is positive.  Since $S_\Bg$
occurs in $\cC_n(B)$ by assumption, and $[S_\Bg]$ is the unique generator
of $\cT$,
Theorem~\ref{th_exp_ind} implies
that
$$
\Prob_{G\in\cC_n(B)}\Bigl[ [S_\Bg]\cap G\ne\emptyset
\Bigr] = 
c_0 + c_1/n + \cdots + c_{r-1}/n^{r-1} + O(1/n^r)
$$
for any $r$, where for $i<r$ we have
$c_i=0$ if $i<\ord(S)$ and $c_i\ne 0$ if $i=\ord(S)$.
Taking any $r>\ord(S)$ yields
$$
\Prob_{G\in\cC_n(B)}\Bigl[ [S_\Bg]\cap G\ne\emptyset
\Bigr] = 
c_i/n^i +  O(1/n^{i+1})
$$
for $i=\ord(S)$ with $c_i>0$, which is bounded below by $C'/n^i$
for any $C'<c_i$ and $n$ sufficiently large.
\end{proof}

\appendix

\providecommand{\bysame}{\leavevmode\hbox to3em{\hrulefill}\thinspace}
\providecommand{\MR}{\relax\ifhmode\unskip\space\fi MR }
\providecommand{\MRhref}[2]{%
  \href{http://www.ams.org/mathscinet-getitem?mr=#1}{#2}
}
\providecommand{\href}[2]{#2}

\end{document}